\documentclass[12pt,letterpaper]{article}
\usepackage[T1]{fontenc}
\usepackage[margin=1in]{geometry} 

\usepackage[titletoc]{appendix}
\usepackage{titling}
\usepackage[authoryear]{natbib}
\usepackage{float}

\usepackage[dvipsnames]{xcolor}
\usepackage[colorlinks=true,
            linkcolor=Blue,
            citecolor=Blue,
            urlcolor=blue, 
            hyperfootnotes=false]{hyperref}
            
\usepackage[bottom]{footmisc}

\usepackage[onehalfspacing]{setspace} 

\usepackage[final]{microtype}

\usepackage[subtle]{savetrees}
\usepackage[small,compact]{titlesec}
\titlespacing{\paragraph}{0pt}{2.25ex plus 1ex minus .2ex}{0.4em}

\usepackage{graphicx, csquotes, enumitem, booktabs, threeparttablex, tabularx, array}
\usepackage{mathtools, amssymb, amsthm, bm, dsfont, etoolbox}
\allowdisplaybreaks

\usepackage{lmodern}

\theoremstyle{definition} 

\newtheorem*{example*}{Example}
\newtheorem{remark}{Remark}
\newtheorem*{remark*}{Remark}

\makeatletter
\newcommand{\exampleqedsymbol}{$\triangle$}

\AtBeginEnvironment{example}{%
  \pushQED{\qed}%
  \let\example@oldsym\qedsymbol
  \renewcommand{\qedsymbol}{\exampleqedsymbol}%
}
\AtEndEnvironment{example}{%
  \popQED%
  \let\qedsymbol\example@oldsym
}

\AtBeginEnvironment{example*}{%
  \pushQED{\qed}%
  \let\example@oldsym\qedsymbol
  \renewcommand{\qedsymbol}{\exampleqedsymbol}%
}
\AtEndEnvironment{example*}{%
  \popQED%
  \let\qedsymbol\example@oldsym
}
\makeatother

\makeatletter
\newcommand{\remarkqedsymbol}{$\triangle$}

\AtBeginEnvironment{remark}{%
  \pushQED{\qed}%
  \let\remark@oldsym\qedsymbol
  \renewcommand{\qedsymbol}{\remarkqedsymbol}%
}
\AtEndEnvironment{remark}{%
  \popQED%
  \let\qedsymbol\remark@oldsym
}

\AtBeginEnvironment{remark*}{%
  \pushQED{\qed}%
  \let\remark@oldsym\qedsymbol
  \renewcommand{\qedsymbol}{\remarkqedsymbol}%
}
\AtEndEnvironment{remark*}{%
  \popQED%
  \let\qedsymbol\remark@oldsym
}
\makeatother

\AtBeginEnvironment{recipe}{\vspace{\topsep}\par\kern8pt\hrule\kern2pt\hrule\relax}

\theoremstyle{plain} 
\newtheorem{assumption}{Assumption}
\newtheorem{proposition}{Proposition}

\newcommand{\R}{\mathbb{R}} 
\renewcommand{\P}[2][\mu]{\mathbb{P}_{#1}\left\{#2\right\}} 
\newcommand{\E}[2][\mu]{\mathbb{E}_{#1}\left[#2\right]} 
\newcommand{\Var}[2][\mu]{\textnormal{Var}_{#1}\left(#2\right)} 
\newcommand{\1}{\bm{1}} 
\newcommand{\I}[1]{\mathds{1}\left\{#1\right\}} 
\newcommand{\abs}[1]{\left|#1\right|} 
\newcommand{\norm}[1]{\left\|#1\right\|} 
\renewcommand{\to}[1][]{\overset{#1}{\rightarrow}} 
\newcommand{\dsqrt}[1]{\displaystyle\sqrt{#1}}

\newcommand{\curly}[1]{\left\{#1\right\}}
\renewcommand{\brack}[1]{\left[#1\right]}
\newcommand{\paren}[1]{\left(#1\right)}

\newcommand{\V}{\mathcal{V}}
\newcommand{\A}{\mathcal{A}}
\newcommand{\M}{\mathcal{M}}
\renewcommand{\l}{\lambda}
\renewcommand{\L}{\Lambda}
\newcommand{\htheta}{\hat{\theta}}

\newcommand{\tSigma}{\widetilde{\Sigma}}
\newcommand{\tsigma}{\widetilde{\sigma}}

\begin{document}
\title{Headline Estimation with Multiple Research Designs\thanks{I thank Isaiah Andrews, Anna Mikusheva, and Alberto Abadie for their guidance and support. I thank seminar participants from the MIT econometrics lunch for helpful comments and discussions. I gratefully acknowledge support from the Jerry A. Hausman Fellowship and the National Science Foundation Graduate Research Fellowship under Grant No. 1745302. I thank ChatGPT and Codex for research assistance.}}
\author{Vod Vilfort\thanks{Department of Economics, Massachusetts Institute of Technology, vod@mit.edu.}}
\date{\today}
\maketitle

\begin{abstract}
\noindent To study a scalar parameter, a researcher may consider multiple research designs. Based on the evidence across designs, the researcher may wish to formulate a headline estimate of the parameter. I examine how to choose this headline when it is unclear which design is most appropriate for studying the parameter. I model this setting by assuming that (i) exactly one of the designs is valid for the parameter and (ii) the researcher has ambiguity about which design is valid, represented by a class of priors over the candidate designs. To account for ambiguity, I propose reporting the headline estimate that minimizes the worst-case posterior risk over the class of priors. In three applications, I show cases where accounting for ambiguity materially affects the quantitative conclusion and cases where an existing headline is already close to optimal.
\end{abstract}

\clearpage

\section{Introduction}
To study a quantity of interest, a researcher may consider multiple research designs. One design might exploit a policy discontinuity, a second might use quasi-random assignment, and a third might leverage comparisons across groups or time periods. When different designs yield similar estimates, the researcher may view this as evidence for a common conclusion about the quantity of interest.\footnote{\citet[page 504]{garin2025impact} say that ``Using multiple research designs allows us to test the sensitivity of our results to empirical strategy.'' \citet[page 2907]{ganong2024spending} say that ``Each of these empirical exercises has distinct advantages and disadvantages, but they all lead to the same conclusion.'' \citet[page 1576]{chen2026women} say that their results ``remain remarkably consistent across both approaches.''} The researcher may then want a \textit{headline estimate}---a single number to emphasize in the abstract, introduction, or policy discussion. I examine how to formulate this headline when it is unclear which design is most appropriate for studying the quantity of interest. 

Researchers consider various approaches to the headline estimation problem. One approach is to headline a precision-weighted average of the design estimates, which gives more weight to designs with smaller standard errors \citep{garin2025impact, jain2026limits}. This can be efficient when every design is unbiased for one common effect, but that assumption is difficult to justify when different designs use distinct sources of identifying variation. A second approach is to headline the estimate from a favored design that is viewed as especially credible or policy relevant \citep{bhuller2017life, sager2025clean, chen2026women, fang2026high}. This avoids the common-effect assumption, but gives no systematic role to other designs. A third approach is to report the range of the design estimates \citep{hastings2018snap, alsan2023civil, ganong2024spending, ng2024returns}. This reveals the set of headlines that are attainable as convex averages of the design estimates, but does not determine what single number should summarize the quantitative conclusion. 

In this paper, I propose a decision-theoretic solution to the headline estimation problem. I model the quantity of interest as a scalar target parameter. Each design produces an estimate of a design-specific estimand, but only one design is \textit{target-valid} in the sense that its estimand identifies the target parameter---i.e., the target-valid design is ``most appropriate'' for studying the quantity of interest. The other estimands retain their design-specific interpretations and may differ from the target, although they can coincide with it numerically. The model therefore allows several ``complementary'' designs to illuminate a shared question without treating their estimates as unbiased measurements of a common effect.\footnote{The framing of designs as being ``complementary'' or as having ``different strengths and weaknesses'' is language that appears frequently in applied work: examples include \citet[page 407]{alsan2023civil}, \citet[page 217]{dunn2024denial}, \citet[page 14]{baran2025clean}, \citet[page 1]{davis2025mobile}, \citet[page 13]{erten2025employment}, \citet[page 3]{han2025trading}, and \citet[page 13]{agrawal2026economics}.} The researcher does not know which design is target-valid, leading to ambiguity over the candidate designs. I model this ambiguity as a class of \textit{design priors}. A design prior specifies how likely each design is to be target-valid. The class collects the set of design priors that the researcher is willing to entertain. Under this structure, I formulate the headline estimation problem as a statistical decision problem under ambiguity \citep{berger1985statistical, gilboa1989maxmin, stoye2012new}. To estimate the target parameter while accounting for ambiguity, I propose a conditional Gamma-minimax procedure \citep{dasgupta1989frequentist, betro1992conditional, giacomini2021robust}. The procedure conditions on the design estimates and chooses the scalar report that minimizes worst-case posterior risk (i.e., expected loss) over the class of design priors. I call this report the \textit{ambiguity-optimal headline estimate}, or \textit{optimal headline} for short.


I propose default specifications for the above decision problem. Under these specifications, the headline estimation procedure requires only the design estimates and their standard errors. The procedure accommodates any compact and convex class of design priors. For example, the \textit{simplex class} is the set of all design priors. Under this class, the procedure chooses the headline that minimizes the largest standardized discrepancy from the design estimates. Equivalently, place an interval around each estimate and expand every interval by the same multiple of its own standard error. The optimal headline is the first point at which all the intervals overlap. With two designs, the optimal headline is an inverse-standard-deviation-weighted average of the two design estimates---which differs from the inverse-variance-weighted average that is used in conventional precision-weighting. With more designs, the solution is determined by the set of pairwise headline estimation problems across designs.

The simplex class represents the benchmark of unrestricted ambiguity. In some cases, the researcher may have more structured forms of ambiguity. For example, the researcher may wish to limit attention to design priors within a neighborhood of a baseline prior. To formalize this, the \textit{baseline-prior perturbation class} begins from a baseline prior and considers the set of priors generated by $\epsilon$-perturbations of the baseline. When the baseline puts equal probability on each design, $\epsilon = 0$ recovers the conventional inverse-variance-weighted average while $\epsilon = 1$ recovers the simplex-optimal headline. Intermediate perturbation levels trace a monotone path between those two endpoint headlines. As a second example, the researcher may have a favored design whose target-validity is deemed to be especially likely. To formalize this, the \textit{favored-design ratio class} considers the set of design priors where the relative probability of the favored design is not too low---governed by a ratio parameter. This formalizes the practice of emphasizing one design by stating how much more plausible it must be than each alternative. As the ratio parameter grows, the optimal headline converges to the favored design estimate. These example structured ambiguity classes can serve as primary specifications when their restrictions are substantively justified or as sensitivity analyses relative to the unrestricted ambiguity benchmark.

In my framework, the performance of a candidate headline is measured by its worst-case posterior risk. Intuitively, the risk quantifies how well the design-specific evidence can support---or be reconciled with---a given headline. By construction, the optimal headline yields the minimum attainable risk. Comparing the risk of a paper's existing headline with the minimum attainable risk shows whether accounting for design ambiguity materially improves the headline report. The headline estimation procedure can therefore be used to support a paper's existing headline or to rank candidate headlines while keeping the degree of reconciliation visible. 

I illustrate my framework in three applications. In \citet{chen2026women}, the target is the effect of courtroom online broadcasting intensity on the gender gap in plaintiff win rate, and the designs are a difference-in-differences (DiD) specification and a Bartik instrumental variable (IV) strategy. Under the simplex class, the optimal headline supports the conclusion that broadcasting narrows the gap, but its high risk suggests that the DiD and IV designs do not tightly reconcile one exact magnitude. Under the favored-design ratio class for DiD---the paper's primary design---I examine how favored DiD would have to be to obtain meaningful risk reductions. In \citet{garin2025impact}, the targets are the effects of a 12-month incarceration sentence on later incarceration and labor market outcomes, and the two designs use sentencing-guideline discontinuities in North Carolina and random judge assignments in Ohio. Under an equal-probability baseline-prior perturbation class, the paper's precision-weighted averages for the main outcomes are already close to optimal, and the paper's substantive conclusions are largely unchanged. In \citet{ganong2024spending}, the target is the marginal propensity to consume out of unemployment benefits, and the six designs exploit benefit-receipt delays, changes in pandemic benefit supplements, and cross-state variation in expiration timing. Under the simplex class, the optimal headline indicates a high spending response, but its high risk reveals that the exact magnitude cannot be tightly reconciled across the six designs.

I emphasize two points regarding the headline estimation framework. First, the framework takes as given that the researcher has specified a scalar target parameter and a set of designs that could plausibly supply the identifying link. It does not decide whether the designs \textit{should} be viewed as speaking to one latent target parameter---that is a substantive judgment for the researcher to make. Second, some settings provide enough economic or structural information to combine different designs to directly identify a target parameter. In those settings, the headline estimation framework is not suitable.\footnote{For example, \citet{laliberte2021long} targets the share of neighborhood effects on long-term educational outcomes that operates through access to better schools, uses a spatial regression discontinuity design to identify the school effects and a movers design to identify the neighborhood effects, then combines them within an education production function decomposition to estimate the target share. The designs are different inputs to a structured identification argument rather than candidate estimates of the target.} 

This paper relates to three literatures. First, robust Bayesian decision theory studies how to make optimal decisions under classes of prior beliefs \citep{gilboa1989maxmin, stoye2012new}. I use this logic to represent ambiguity over the identity of the target-valid design and to choose a headline that minimizes worst-case posterior risk. The conditional Gamma-minimax criterion is applied conditional on the observed estimates, which distinguishes it from ex-ante Gamma-minimax criteria that evaluate decision rules based on worst-case Bayes risk. Advantages of the conditional approach include analytical and numerical tractability \citep{christensen2026optimal} and favorable dynamic consistency properties \citep{lim2026dynamically}. A disadvantage is that the resulting decision rules are not guaranteed to be admissible \citep{giacomini2021robust}.

Second, work on averaging and adaptation examines how to use a precise estimator that may be biased together with a robust estimator that is valid under weaker assumptions \citep{green1991james, cheng2019uniform, armstrong2025adapting}. Here the use of a risk-based performance criterion is related, but the present problem has no estimator known in advance to provide an unbiased anchor. Which design is target-valid is precisely what remains ambiguous.

Third, this paper contributes to recent work on how to combine evidence across models or research designs. \citet{bhattacharya2026robust} construct data-dependent weights for methodological triangulation using testable implications of candidate causal models. \citet{park2026choosing} exploit an ordering among matching, DiD, and hybrid estimands to select the hybrid estimand under minimax regret. My framework instead chooses a general scalar report while design ambiguity remains unresolved.\footnote{There is also more classical work on frequentist and Bayesian model averaging \citep{hansen2007least, hoeting1999bayesian}. With a singleton design-prior class, the optimal headline reduces to a Bayes estimator that aligns with the logic of Bayesian model averaging estimators.} \citet{vilfort2026robust} develops inference procedures for settings where there is ambiguity over how to average causal effects from a given design---e.g., an event study design may generate many causal effects across different treatment cohorts. The present framework instead considers decision-making in settings where there is ambiguity over a collection of candidate designs---e.g., a researcher may consider a weighted average for an event study design and a weighted average for an IV design.

The remainder of this paper proceeds as follows. Section \ref{arxiv1:sec:example.setting} illustrates the headline estimation framework in an example setting. Section \ref{arxiv1:sec:decision.problem} develops the general decision problem and imposes my default specifications. Section \ref{arxiv1:sec:proposed.procedures} characterizes the optimal headline under general ambiguity classes. Section \ref{arxiv1:sec:asymptotic.results} gives asymptotic results. Section \ref{arxiv1:sec:empirical.applications} presents empirical applications. Section \ref{arxiv1:sec:conclusion} concludes. The supplemental appendix contains proofs.

\section{Example Setting}\label{arxiv1:sec:example.setting}
In this section, I illustrate the framework at a high level in a two-design setting, postponing the general formulation and additional sensitivity analysis to later sections. \citet{chen2026women} study a judicial reform in China that mandated courts to broadcast legal proceedings on an online platform, and examine how changes in broadcasting intensity affect the female-male plaintiff win-rate gap---estimated to be $-0.0356$ in their all-litigants sample.\footnote{This is the estimated female coefficient from the authors' pre-reform all-litigants regression of win rate on a female indicator, control variables, and fixed effects \citep[Table 1]{chen2026women}.} Let $\theta \in \R$ denote the target parameter, defined as ``the effect of broadcasting intensity on the gender gap in plaintiff win rate.'' The authors' primary design uses a DiD specification, while their alternative design uses a Bartik IV to instrument for broadcasting intensity. For the all-litigants sample, the reported estimates and standard errors are
\begin{align}\label{arxiv1:eq:example.setting.evidence}
    Y_{\mathrm{DiD}} = 0.0394, \quad \sigma_{\mathrm{DiD}} = 0.00255, \quad Y_{\mathrm{IV}} = 0.0726, \quad \sigma_{\mathrm{IV}} = 0.00599.
\end{align}
This means that a 10 percentage-point increase in broadcasting intensity corresponds to a 0.394 percentage-point narrowing of the gender gap under the DiD design, and a 0.726 percentage-point narrowing under the IV design. Both designs imply that greater broadcasting intensity narrows the gender gap, although the IV estimate is larger. Suppose that a researcher wishes to formulate a headline estimate of $\theta$, but is unsure which of the two designs provides the more appropriate link to $\theta$.

Let design $1$ denote DiD and design $2$ denote IV. Let $(Y_{1},Y_{2})$ denote the design estimates, $(\sigma_{1},\sigma_{2})$ their standard errors, and $(\mu_{1},\mu_{2})$ their unknown estimands. One design is \textit{target-valid}: for some unknown $k \in \{1,2\}$, the target parameter is identified as $\theta = \mu_{k}$. To estimate $\theta$, the researcher must account for two unknowns: the value of the estimands $(\mu_{1},\mu_{2})$ and the identity of the target-valid design $k$. 

To account for the unknown estimands $(\mu_{1}, \mu_{2})$, the researcher conditions on the observed estimates $(Y_{1}, Y_{2})$ and adopts a diffuse-limit normal posterior, which yields $\mu_{1}|Y \sim N(Y_{1},\sigma_{1}^{2})$ and $\mu_{2}|Y \sim N(Y_{2},\sigma_{2}^{2})$. Let $a \in \R$ denote a candidate headline for $\theta$. If $k$ is target-valid, then the posterior mean squared error (MSE) is
\begin{align*}
    \underbrace{\E[]{(\theta - a)^{2}|Y,k}}_{\text{posterior MSE}} = \underbrace{(Y_{k} - a)^{2}}_{\substack{\text{posterior} \\ \text{squared bias}}} + \underbrace{\sigma_{k}^{2}}_{\substack{\text{posterior} \\ \text{variance}}}, \qquad \theta|Y,k \overset{d}{=} \mu_{k}|Y,k \sim \underbrace{N(Y_{k}, \sigma_{k}^{2})}_{\substack{\text{posterior} \\ \text{distribution}}}.
\end{align*}
Consider an oracle that knows $k$ and therefore headlines the design-$k$ estimate $a^{*} = Y_{k}$, leading to posterior MSE $\sigma_{k}^{2}$. Under $k$, the posterior MSE of a candidate headline $a \in \R$ relative to the oracle benchmark is given by the design-$k$ posterior risk
\begin{align*}
    R_{k}(a|Y) = \frac{\E[]{(\theta - a)^{2}|Y,k}}{\sigma_{k}^{2}} = \frac{(Y_{k} - a)^{2}}{\sigma_{k}^{2}} + 1, \quad k \in \{1,2\}.
\end{align*}
This risk quantity measures the posterior mean-squared discrepancy between $\theta$ and $a$ relative to how precisely design $k$ can be estimated---I show in Section \ref{arxiv1:sec:decision.problem} that $R_{k}(a|Y)$ arises naturally from a decision problem with loss function $L(a,\mu,k) = (\mu_{k} - a)^{2}/\sigma_{k}^{2}$.

To account for the unknown target-valid design $k$, the researcher considers different prior beliefs over $k$, which in this setting correspond to different prior probabilities $\l \in [0,1]$ that DiD is the target-valid design. Under a given design prior $\l$, the posterior risk is
\begin{align*}
    R_{\l}(a|Y) = \l R_{1}(a|Y) + (1-\l) R_{2}(a|Y).
\end{align*}
Each design prior $\l$ reflects \textit{uncertainty} over research designs. For example, $\l = 2/3$ says that DiD is twice as likely to be target-valid as IV, while $\l = 1/3$ says that DiD is half as likely. The researcher may entertain many such priors, giving rise to a class of design priors $\L \subseteq [0,1]$. The class $\L$ reflects \textit{ambiguity} over research designs.\footnote{This distinction between uncertainty and ambiguity---a single belief versus a class of beliefs---is consistent with standard formulations in decision theory \citep{berger1985statistical, gilboa1989maxmin, stoye2012new}.} The optimal headline is the action $a = \htheta_{\L}^{*}(Y)$ that minimizes worst-case posterior risk over the class $\L$. In particular, the optimal headline and the optimized risk are
\begin{align*}
    \htheta_{\L}^{*}(Y) = \arg\min_{a \in \R}\max_{\l \in \L}R_{\l}(a|Y), \quad R_{\L}^{*}(Y) = \min_{a \in \R}\max_{\l \in \L}R_{\l}(a|Y).
\end{align*}
Intuitively, the optimal headline $\htheta_{\L}^{*}(Y)$ guards against the worst-case design priors in $\L$. The worst-case posterior risk $\max_{\l \in \L}R_{\l}(a|Y)$ measures the consequences of compressing design-specific evidence into a scalar headline $a$, and the optimized risk $R_{\L}^{*}(Y) = \max_{\l \in \L}R_{\l}(\htheta_{\L}^{*}(Y)|Y)$ is the minimum attainable risk. The excess risk $R_{\L}^{*}(Y) - 1$ gives the proportional increase in risk relative to the oracle that knows the target-valid design, which provides a natural benchmark for interpreting the magnitudes of $R_{\L}^{*}(Y)$. 

The simplex class $\L_{\Delta} = [0,1]$ represents the benchmark of unrestricted ambiguity, wherein the researcher entertains every possible design prior. Under the simplex class, the solution is
\begin{align*}
    \htheta_{\Delta}^{*}(Y)
    &= \frac{\sigma_{1}^{-1}}{\sigma_{1}^{-1} + \sigma_{2}^{-1}}Y_{1} + \frac{\sigma_{2}^{-1}}{\sigma_{1}^{-1} + \sigma_{2}^{-1}}Y_{2}, \quad
    R_{\Delta}^{*}(Y) = \paren{\frac{\abs{Y_{1} - Y_{2}}}{\sigma_{1} + \sigma_{2}}}^{2} + 1.
\end{align*}
The optimal headline $\htheta_{\Delta}^{*}(Y)$ is an inverse-standard-deviation-weighted average. It places higher weight on the more precise design, but less aggressively than the conventional inverse-variance-weighted average considered in models where both designs are independent unbiased estimates of a common parameter. The optimized risk $R_{\Delta}^{*}(Y)$ quantifies how well the DiD and IV designs can support---or be reconciled with---the optimal headline. For the simplex class, the square-root excess risk $c_{\Delta}^{*}(Y) = \sqrt{R_{\Delta}^{*}(Y) - 1}$ provides a natural measure of reconciliation cost. In particular, $c_{\Delta}^{*}(Y)$ is the smallest standard-error multiplier needed for two-sided intervals around the design estimates to share a common point:
\begin{align*}
    c_{\Delta}^{*}(Y) = \min\curly{c \geq 0: I_{1}(c) \cap I_{2}(c) \neq \varnothing}, \quad I_{k}(c) = [Y_{k}-c\sigma_{k}, Y_{k}+c\sigma_{k}], \quad k \in \{1,2\}.
\end{align*}
That first common point is the simplex-optimal headline: $I_{1}(c_{\Delta}^{*}(Y)) \cap I_{2}(c_{\Delta}^{*}(Y)) = \curly{\htheta_{\Delta}^{*}(Y)}$.

From the \citet{chen2026women} DiD and IV design estimates and standard errors in \eqref{arxiv1:eq:example.setting.evidence}, the optimal headline and optimized risk under the simplex class $\L_{\Delta} = [0,1]$ are 
\begin{align*}
    \htheta_{\Delta}^{*}(Y) = 0.0493, \quad R_{\Delta}^{*}(Y) = 16.11.
\end{align*}
In particular, a 10 percentage-point increase in broadcasting intensity yields a 0.493 percentage-point narrowing of the gender gap in plaintiff win rate. Headlining the DiD estimate $a = Y_{1}$ yields worst-case risk $R_{\Delta}(Y_{1}|Y) = 31.72$, so the optimal headline lowers risk by $(1-(16.11/31.72)) \times 100 \approx 49$ percent. Even so, the optimized risk $R_{\Delta}^{*}(Y) = 16.11$ yields a $(16.11 - 1) \times 100 \approx 1500$ percent increase relative to the oracle that knows which design is target-valid. In particular, for the design intervals $I_{k}(c) = [Y_{k} - c\sigma_{k}, Y_{k} + c\sigma_{k}]$ to share a common point, one requires a standard-error multiplier of at least $c_{\Delta}^{*}(Y) = \dsqrt{16.11 - 1} \approx 3.89$. Thus, while the optimal headline supports the conclusion that broadcasting narrows the gender gap, the high risk suggests that the DiD and IV designs do not tightly reconcile one exact magnitude. When I return to this application in Section \ref{arxiv1:sec:empirical.applications}, I will consider a favored-design ratio class that privileges DiD---the paper's primary design---and examine how favored DiD would have to be to obtain meaningful
risk reductions.  

The above analysis considered a two-design setting with an unrestricted ambiguity class. The general decision problem allows for more designs and other ambiguity classes. The exposition below will separately introduce the ingredients used in the above analysis, including the latent-design setup, the beliefs over estimands and designs, the loss function preferences, and the proposed optimality criterion.

\section{Decision Problem}\label{arxiv1:sec:decision.problem}
Consider a researcher who wishes to estimate a scalar target parameter $\theta \in \R$. The researcher observes estimates $Y_{k} \in \R$ from research designs $k \in [K] = \{1,\ldots,K\}$, where $K \geq 2$. Letting $Y = (Y_{1}, \ldots, Y_{K})'$ denote the vector of design estimates, I assume that $Y$ is normally distributed:
\begin{align*}
    Y \sim N(\mu,\Sigma),
\end{align*}
where $\mu = (\mu_{1},\ldots,\mu_{K})' \in \M \subseteq \R^{K}$ is a vector of unknown design estimands, $\mu_{k} = \E{Y_{k}}$, and $\Sigma$ is a known positive definite covariance matrix. This assumption is motivated by large-sample approximations formalized in Section \ref{arxiv1:sec:asymptotic.results}. Let $\sigma_{k}^{2} = \Sigma_{kk} = \Var{Y_{k}} > 0$.

I represent the set of research designs with the set of standard unit vectors $\V = \{v_{1}, \ldots, v_{K}\}$. In particular, a given design $v_{k} \in \V$ yields $(v_{k}'Y, v_{k}'\mu, v_{k}'\Sigma v_{k}) = (Y_{k}, \mu_{k}, \sigma_{k}^{2})$. For some latent design $v \in \V$, the target parameter $\theta$ is identified by the corresponding design estimand: $\theta = v'\mu$. I refer to $v$ as the \textit{target-valid} design. Thus, the underlying state variable is the unknown estimand-design pair $(\mu, v) \in \M \times \V$, while the target is a state-dependent functional $\theta(\mu, v) = v'\mu$. For example, $v = v_{k}$ means that design $k$ is the target-valid design: $\theta = \theta(\mu, v_{k}) = \mu_{k}$. No restriction is imposed on design estimands $\mu_{j}$ where $v_{j} \neq v$. Those estimands may differ from $\theta$ by arbitrary amounts, although they may coincide with $\theta$ numerically at some realizations of $\mu \in \M$. Thus, the value of $v$ concerns the identifying link between a design and the target parameter---not a numerical requirement that exactly one coordinate of $\mu$ equal $\theta$.

The researcher has a loss function $L(a,\mu,v) \geq 0$ that quantifies the consequences of taking action $a \in \A \subseteq \R$ when the estimand is $\mu \in \M$ and the target-valid design is $v \in \V$. For example, the squared error loss function $L(a,\mu,v) = (v'\mu - a)^{2}$ measures how far the action $a$ is from the parameter $v'\mu = \theta$ when the estimand-design pair is $(\mu,v)$. In other words, it measures the error relative to whichever design is target-valid in the state. I consider a general class of such loss functions in Section \ref{arxiv1:sec:default.specifications}.

\subsection{Researcher Beliefs}\label{arxiv1:sec:researcher.beliefs}
Let $\Delta(\mathcal{X})$ represent the set of probability distributions over a given set $\mathcal{X}$. Let $\gamma \in \Delta(\M \times \V)$ denote prior beliefs about the state $(\mu,v)$. This joint prior can be decomposed as
\begin{align*}
    \gamma(d\mu,v_{k}) = \pi_{k}(d\mu)\l_{k},  \quad \pi_{k}(\cdot) = \pi(\cdot|v=v_{k}), \quad \l_{k} = \P[\l]{v = v_{k}},
\end{align*}
where $\pi_{k} \in \Delta(\M)$ is the conditional prior on $\mu \in \M$ given $v = v_{k}$ and $\l_{k}$ is the marginal prior probability that $v = v_{k}$. I call the vector $\l = (\l_{1}, \ldots, \l_{K})' \in \Delta(\V)$ a \textit{design prior}; its components are the design probabilities. Here, $\Delta(\V)$ denotes the standard simplex in $\R^{K}$:
\begin{align*}
    \Delta(\V) = \curly{\l \in \R^{K}: \l \geq 0,\ \sum_{k = 1}^{K}\l_{k} = 1}.
\end{align*}
The researcher has ambiguity over research designs, represented by a nonempty, compact, and convex class $\L \subseteq \Delta(\V)$ of design priors $\l \in \L$. Below are three example classes.

\begin{example*}[Simplex class]
The unrestricted-ambiguity benchmark is the simplex class:
\begin{align}\label{arxiv1:eq:simplex.class}
    \L_{\Delta} = \Delta(\V) = \curly{\l \in \R^{K}: \l \geq 0,\ \sum_{k = 1}^{K}\l_{k} = 1}.
\end{align}
Under the simplex class, the researcher entertains every belief about $v$, including dogmatic priors (the simplex vertices $\V = \{v_{1}, \ldots, v_{K}\}$) that put probability one on a single design.
\end{example*}

\begin{example*}[Baseline-prior perturbation class]
Consider a baseline design prior $\l^{0} \in \Delta(\V)$. The perturbation class is the set of $\epsilon$-perturbations towards other priors $q \in \Delta(\V)$ in the simplex:
\begin{align}\label{arxiv1:eq:perturbation.class}
    \L_{\epsilon}(\l^{0}) = (1 - \epsilon)\l^{0} + \epsilon \Delta(\V) = \curly{(1 - \epsilon)\l^{0} + \epsilon q: q \in \Delta(\V)}, \quad \epsilon \in [0,1].
\end{align}
The baseline $\l^{0}$ represents the researcher's best assessment about the uncertainty in $v$, while the perturbation level $\epsilon$ measures the ambiguity in that assessment. The perturbation class reduces to $\L_{0}(\l^{0}) = \{\l^{0}\}$ at $\epsilon = 0$ and expands to $\L_{1}(\l^{0}) = \Delta(\V)$ at $\epsilon = 1$.
\end{example*}

\begin{example*}[Favored-design ratio class]
Suppose a given design is favored in the sense that it is believed to be at least $r \geq 1$ times as plausible as every other design. Without loss of generality, suppose that $k = 1$ is the favored design. The ratio class is defined as
\begin{align}\label{arxiv1:eq:ratio.class}
    \L_{r} = \curly{\l \in \Delta(\V): \l_{1} \geq r\l_{k},\ \forall k \neq 1}, \quad r \geq 1.
\end{align}
The ratio parameter $r$ places a lower bound on the favored design's relative probability of being target-valid. If $K = 2$, the ratio class can be expressed as $\L_{r} = \{\l: \l_{1} \geq r/(r+1)\}$, which says that the researcher entertains any design prior satisfying $\P[\l]{v = v_{1}} \geq r/(r+1)$.
\end{example*}

Combining the conditional $\mu$-priors with a class of design priors yields a class of joint priors:
\begin{align*}
    \Gamma = \curly{\gamma: \gamma(d\mu, v_{k}) = \pi_{k}(d\mu)\l_{k},\ \forall k \in [K],\ \l \in \L}.
\end{align*}
Let $\phi_{\Sigma}(\cdot|\mu)$ denote the density function of $N(\mu,\Sigma)$. The marginal likelihood of $Y$ conditional on $v = v_{k}$ is
\begin{align*}
    m_{\pi, k}(Y) = \int_{\M} \phi_{\Sigma}(Y|\mu) \pi_{k}(d\mu).
\end{align*}
I assume that $m_{\pi, k}(Y) > 0$ for all $k \in [K]$. Bayes' rule gives the posterior beliefs
\begin{align*}
    \gamma(d\mu,v_{k}|Y) = \pi_{k}(d\mu|Y)\l_{\pi, k}(Y), \quad \pi_{k}(d\mu|Y) = \frac{\phi_{\Sigma}(Y|\mu)\pi_{k}(d\mu)}{m_{\pi, k}(Y)}, \quad \l_{\pi, k}(Y) = \frac{m_{\pi, k}(Y)\l_{k}}{\sum_{j = 1}^{K}m_{\pi, j}(Y)\l_{j}}.
\end{align*}
That is, the prior beliefs $(\gamma, \pi_{k})$ are updated into posterior beliefs $(\gamma(\cdot|Y), \pi_{k}(\cdot|Y))$. Likewise, $\l_{\pi, k}(Y)$ is the posterior probability that $v = v_{k}$. Letting $\l_{\pi}(Y) = (\l_{\pi, 1}(Y), \ldots, \l_{\pi, K}(Y))'$ denote the \textit{design posterior} associated with design prior $\l$, the \textit{design posterior class} is 
\begin{align*}
    \L_{\pi}(Y) = \curly{\l_{\pi}(Y): \l \in \L}.
\end{align*}
For $\L_{\Delta} = \Delta(\V)$, the posterior simplex class preserves unrestricted ambiguity: $\L_{\pi,\Delta}(Y) = \L_{\Delta}$. For $\L_{\epsilon}(\l^{0}) = (1 - \epsilon)\l^{0} + \epsilon \Delta(\V) = \{\l: \l \geq (1 - \epsilon)\l^{0}\}$, the posterior perturbation class is
\begin{align*}
    \L_{\pi,\epsilon}(Y;\l^{0}) = \curly{\Bar{\l} \in \Delta(\V): \frac{\Bar{\l}_{k}/m_{\pi,k}(Y)}{\sum_{j = 1}^{K} \Bar{\l}_{j}/m_{\pi,j}(Y)} \geq (1 - \epsilon)\l_{k}^{0},\ \forall k}.
\end{align*}
For $\L_{r} = \{\l: \l_{1} \geq r\l_{k},\ \forall k \neq 1\}$, the posterior ratio class is
\begin{align*}
    \L_{\pi,r}(Y) = \curly{\Bar{\l} \in \Delta(\V): \frac{\Bar{\l}_{1}/m_{\pi,1}(Y)}{\sum_{j = 1}^{K} \Bar{\l}_{j}/m_{\pi,j}(Y)} \geq r\frac{\Bar{\l}_{k}/m_{\pi,k}(Y)}{\sum_{j = 1}^{K} \Bar{\l}_{j}/m_{\pi,j}(Y)},\ \forall k \neq 1}.
\end{align*}
If the conditional $\mu$-priors $\pi_{k}$ are common across designs $k$, then marginal likelihoods $m_{\pi,k}(Y)$ are common and any design posterior class coincides with its design prior class: $\L_{\pi}(Y) = \L$.

\subsection{Optimality Criterion}
To estimate the target parameter $\theta$ while accounting for ambiguity over designs, I propose a conditional Gamma-minimax approach \citep{dasgupta1989frequentist, betro1992conditional, giacomini2021robust}. For each $k \in [K]$, define the design-specific posterior risk of action $a \in \A \subseteq \R$ as the posterior expected loss conditional on $v = v_{k}$:
\begin{align*}
    R_{\pi,k}(a|Y) = \int_{\M} L(a,\mu,v_{k}) \pi_{k}(d\mu|Y), \quad v = v_{k}.
\end{align*}
Under design prior $\l \in \L$, the posterior risk is
\begin{align*}
    R_{\pi,\l}(a|Y) = \sum_{k = 1}^{K} \l_{\pi, k}(Y) R_{\pi,k}(a|Y).
\end{align*}
For candidate headline $\htheta(Y)$, evaluate the design-specific and posterior risks at action $a = \htheta(Y)$. The conditional Gamma-minimax action $\htheta_{\pi,\L}^{*}(Y)$ minimizes the worst-case posterior risk over the class of joint priors $\Gamma$, and thus equivalently over the class of design priors $\L$:
\begin{align*}
    \htheta_{\pi,\L}^{*}(Y) \in \arg\min_{a \in \A} \max_{\l \in \L} R_{\pi,\l}(a|Y) = \arg\min_{a \in \A} \max_{\Bar{\l} \in \L_{\pi}(Y)} \sum_{k = 1}^{K} \Bar{\l}_{k}R_{\pi,k}(a|Y).
\end{align*}
The minimized worst-case posterior risk is
\begin{align*}
    R_{\pi,\L}^{*}(Y) = \max_{\l \in \L} R_{\pi,\l}(\htheta_{\pi,\L}^{*}(Y)|Y) = \max_{\Bar{\l} \in \L_{\pi}(Y)}\sum_{k = 1}^{K} \Bar{\l}_{k} R_{\pi,k}(\htheta_{\pi,\L}^{*}(Y)|Y).
\end{align*}
Thus, the conditional Gamma-minimax criterion defines the optimal headline as an optimal action $a=\htheta_{\pi,\L}^{*}(Y)$ under design ambiguity, where the maximization over $\l \in \L$ (or $\Bar{\l} \in \L_{\pi}(Y)$) accounts for the design ambiguity. The optimized risk $R_{\pi,\L}^{*}(Y)$ measures the consequences of the optimal action under worst-case design beliefs. I refer to $\htheta_{\pi,\L}^{*}(Y)$ as the \textit{ambiguity-optimal headline estimate}, or \textit{optimal headline} for short.\footnote{For a singleton $\L = \{\l\}$, the optimal headline is a Bayes estimator under the joint prior $\gamma$ induced by $(\pi, \l)$.}

\subsection{Default Specifications}\label{arxiv1:sec:default.specifications}
The above decision problem allows for general specifications of the researcher's loss function and beliefs. In this section, I propose default specifications which yield particularly tractable headline estimation procedures. I characterize the resulting procedures in Section \ref{arxiv1:sec:proposed.procedures}.

\subsubsection{Loss Function}
Let $\A = \R$. I consider weighted squared error loss functions
\begin{align*}
    L(a,\mu,v) = w(v)(v'\mu - a)^{2}, \quad w(v) > 0.
\end{align*}
The squared error component $(v'\mu - a)^{2}$ measures the discrepancy between the action $a \in \R$ and the target parameter $\theta = v'\mu$ when the estimand-design pair is $(\mu,v)$. The loss weights $w(v)$ allow the penalty from such discrepancies to depend on the latent design $v$. Letting $w_{k} = w(v_{k})$, the design-specific posterior risk is a conditional MSE:
\begin{align*}
    R_{\pi,k}(a|Y) = w_{k}\curly{\paren{\E[\pi]{\mu_{k}|v = v_{k},Y} - a}^{2} + \Var[\pi]{\mu_{k}|v = v_{k},Y}}, \quad v = v_{k}.
\end{align*}
For design prior $\l \in \L$, the posterior risk is therefore an average conditional MSE:
\begin{align*}
    R_{\pi,\l}(a|Y) = \sum_{k = 1}^{K} \l_{\pi, k}(Y) w_{k}\curly{\paren{\E[\pi]{\mu_{k}|v = v_{k},Y} - a}^{2} + \Var[\pi]{\mu_{k}|v = v_{k},Y}}.
\end{align*}
The loss weights $w_{k}$ play a different role from the class $\L$. The loss weights describe the cost of errors given the target-valid design, whereas the class describes the ambiguity over which design is target-valid. The former captures preferences while the latter captures beliefs. Either component can be varied while holding the other fixed. 

\subsubsection{Beliefs over Estimands} 
Let $\M = \R^{K}$ and consider common normal conditional $\mu$-priors:
\begin{align*}
    \pi_{k} = \pi_{\varsigma} = N(\beta_{0}, \varsigma^{-1}\Omega_{0}), \quad \forall k \in [K],
\end{align*}
where $\Omega_{0}$ is positive definite and $\varsigma > 0$. Because the conditional $\mu$-priors are common, so are the marginal likelihoods $m_{\pi, k}(Y) = m_{\pi}(Y)$, and hence any design posterior class coincides with its design prior class: $\L_{\pi}(Y) = \L$. The conditional $\mu$-posteriors are also common:
\begin{align*}
    \pi_{k}(\cdot|Y) = \pi_{\varsigma}(\cdot|Y) = N\paren{\paren{\varsigma\Omega_{0}^{-1}+\Sigma^{-1}}^{-1}\paren{\varsigma\Omega_{0}^{-1}\beta_{0}+\Sigma^{-1}Y}, \paren{\varsigma\Omega_{0}^{-1}+\Sigma^{-1}}^{-1}}, \quad \forall k \in [K].
\end{align*}
Letting $\varsigma \to 0$ yields a convenient diffuse-limit normal posterior 
\begin{align*}
    \pi_{\varsigma}(\cdot|Y) \to[d] \pi_{0}(\cdot|Y) = N(Y,\Sigma).
\end{align*}
This approximates settings where the researcher has $\mu$-priors that are low precision relative to the information in $Y$. Here this takes the form of a diffuse normal prior sequence $\{\pi_{\varsigma}\}_{\varsigma \to 0}$. I use $\pi_{0}(\cdot|Y) = N(Y,\Sigma)$ as the default $\mu$-posterior. The design-specific posterior risk becomes
\begin{align*}
    R_{\pi_{0},k}(a|Y) = R_{k}(a|Y) = w_{k}\curly{(Y_{k} - a)^{2} + \sigma_{k}^{2}}, \quad v = v_{k}.
\end{align*}
The posterior risk under $\l \in \L$ is therefore
\begin{align*}
    R_{\pi_{0},\l}(a|Y) = R_{\l}(a|Y) = \sum_{k = 1}^{K}\l_{k}w_{k}\curly{(Y_{k} - a)^{2} + \sigma_{k}^{2}}, \quad \l_{\pi_{0},k}(Y) = \l_{k}.
\end{align*}
Because design posteriors $\l_{\pi_{0}}(Y)$ coincide with design priors $\l$ under the default specifications, I will simply refer to the latter when speaking about design beliefs.

\subsubsection{Loss Function Weights}
I consider precision weights
\begin{align*}
    w(v) = (v'\Sigma v)^{-1} > 0, \quad w_{k} = w(v_{k}) = 1/\sigma_{k}^{2}.
\end{align*}
The resulting standardized squared error loss function is
\begin{align*}
    L(a,\mu,v) = \frac{(v'\mu - a)^{2}}{v'\Sigma v}.
\end{align*}
This places the design-specific errors on a common standard-error scale and prevents a highly imprecise design estimate from dominating the decision problem. The precision weights also yield the following design-neutrality property: an oracle that observes $v$ before acting incurs the same posterior risk regardless of the value of $v$. Intuitively, the oracle is not worse off learning that design $k$ is target-valid versus learning that design $j$ is target-valid.

\begin{proposition}[Design-neutrality]\label{arxiv1:prop:precision.weights}
The oracle procedure that reports $Y_{k}$ if and only if design $k$ is target-valid yields oracle posterior risk
\begin{align*}
    \sum_{k = 1}^{K}\I{v = v_{k}}R_{k}(Y_{k}|Y) = \sum_{k = 1}^{K}\I{v = v_{k}}w_{k}\sigma_{k}^{2}, \quad v \in \V.
\end{align*}
The oracle risk is the same across realizations of $v$ if and only if $w_{k} = c/\sigma_{k}^{2}$ for a constant $c > 0$. Normalizing to $c = 1$ yields the precision weights $w_{k} = 1/\sigma_{k}^{2}$.
\end{proposition}

\begin{proof}
See Appendix \ref{arxiv1:app:proof:precision.weights}.    
\end{proof}

Under precision weights $w_{k} = 1/\sigma_{k}^{2}$, the design-specific posterior risk becomes
\begin{align*}
   R_{k}(a|Y) = \frac{(Y_{k} - a)^{2}}{\sigma_{k}^{2}} + 1, \quad v = v_{k}.
\end{align*}
In words, the design-specific posterior risk of an action $a \in \R$ is equal to one plus the squared standardized discrepancy between $a$ and $Y_{k}$. Intuitively, the design-$k$ penalty of headlining $a$ grows with $a$'s distance from the design estimate relative to that design's standard error. The posterior risk under $\l \in \L$ is the corresponding average squared standardized discrepancy ($+1$) across designs: 
\begin{align*}
    R_{\l}(a|Y) = \sum_{k = 1}^{K}\l_{k}\frac{(Y_{k} - a)^{2}}{\sigma_{k}^{2}} + 1.
\end{align*}
Thus, $R_{\l}(a|Y)-1$ gives the mean-squared standardized discrepancy of a candidate headline $a$ from the design estimates $(Y_{1}, \ldots, Y_{K})$, where the mean is taken under $\l \in \L$.

\begin{remark}[Relationship to unweighted squared error loss]
For unweighted (i.e., $w(v) = 1$) squared error loss $L(a,\mu,v) = (v
'\mu - a)^{2}$, the design-specific posterior risk at $v_{k}$ is $(Y_{k} - a)^{2} + \sigma_{k}^{2}$. Comparing this posterior risk to that of an oracle with knowledge of $v$ yields posterior regret
\begin{align*}
    \frac{(Y_{k} - a)^{2} + \sigma_{k}^{2}}{\min_{a \in \R}\curly{(Y_{k} - a)^{2} + \sigma_{k}^{2}}} = \frac{(Y_{k} - a)^{2}}{\sigma_{k}^{2}} + 1, \quad v = v_{k}.
\end{align*}
This posterior regret measure is analogous to the one considered in \citet{andrews2026communicating}, but defined here as a ratio of posterior risks instead of a level-difference. The ratio comparison is analogous to the ex-ante adaptation regret formulation considered in \citet{armstrong2025adapting}. Thus, the posterior risk under precision-weighted squared error loss coincides with a measure of posterior (ex-post) adaptation regret under unweighted squared error loss.
\end{remark}

\section{Proposed Procedures}\label{arxiv1:sec:proposed.procedures}
Under the default specifications in Section \ref{arxiv1:sec:default.specifications}, the optimal headline estimate is
\begin{align*}
    \htheta_{\L}^{*}(Y) = \arg\min_{a \in \R} \max_{\l \in \L} \sum_{k = 1}^{K} \l_{k}\frac{(Y_{k} - a)^{2}}{\sigma_{k}^{2}} + 1.
\end{align*}
For any candidate headline $a \in \R$, its worst-case posterior risk is
\begin{align*}
    R_{\L}(a|Y) = \max_{\l \in \L} \sum_{k = 1}^{K} \l_{k}\frac{(Y_{k} - a)^{2}}{\sigma_{k}^{2}} + 1.
\end{align*}
The $+1$ term is the irreducible oracle risk of the infeasible procedure that reports $Y_{k}$ if and only if design $k$ is target-valid. Thus, $(R_{\L}(a|Y) - 1) \times 100$ is the percentage increase over oracle risk. The optimal headline and optimized risk are
\begin{align*}
    \htheta_{\L}^{*}(Y) = \arg\min_{a \in \R} R_{\L}(a|Y), \quad
    R_{\L}^{*}(Y) = \min_{a \in \R} R_{\L}(a|Y) = R_{\L}(\htheta_{\L}^{*}(Y)|Y).
\end{align*}
If a paper has a baseline headline $a_{B} \in \R$, then it is useful to define
\begin{align*}
    R_{\L}^{B} = R_{\L}(a_{B}|Y), \quad R_{\L}^{*} = R_{\L}(\htheta_{\L}^{*}(Y)|Y).
\end{align*}
Comparing the baseline risk $R_{\L}^{B}$ with the optimized risk $R_{\L}^{*}$ shows whether accounting for the design ambiguity represented by $\L$ materially improves the headline report. The reduction in worst-case risk is $R_{\L}^{B} - R_{\L}^{*}$, and the percentage reduction is
\begin{align*}
    \paren{1 - \frac{R_{\L}^{*}}{R_{\L}^{B}}} \times 100.
\end{align*}
This provides a convenient bounded summary.

In Section \ref{arxiv1:sec:general.characterization}, I provide a general characterization of the proposed procedure for nonempty, compact, and convex classes $\L$. In Section \ref{arxiv1:sec:characterization.full.ambiguity}, I specialize to the simplex class $\L_{\Delta}$ defined in \eqref{arxiv1:eq:simplex.class}, which represents the benchmark of unrestricted ambiguity. In Section \ref{arxiv1:sec:characterization.structured.ambiguity}, I consider more structured forms of ambiguity, specializing first to the perturbation class $\L_{\epsilon}(\l^{0})$ defined in \eqref{arxiv1:eq:perturbation.class} and then to the ratio class $\L_{r}$ defined in \eqref{arxiv1:eq:ratio.class}.

\subsection{General Characterization}\label{arxiv1:sec:general.characterization}
Let $D_{k}(a|Y) = \abs{Y_{k} - a}/\sigma_{k}$ denote the standardized discrepancy under design $k$, which measures the distance between headline $a$ and estimate $Y_{k}$ in units of the design's standard error. For any nonempty, compact, and convex class $\L \subseteq \Delta(\V)$, the optimal headline and optimized risk are
\begin{align*}
    \htheta_{\L}^{*}(Y) = \arg\min_{a \in \R}R_{\L}(a|Y), \quad R_{\L}^{*}(Y) = R_{\L}(\htheta_{\L}^{*}(Y)|Y), \quad R_{\L}(a|Y) = \max_{\l \in \L} \sum_{k = 1}^{K} \l_{k}D_{k}(a|Y)^{2} + 1.
\end{align*}
For any design prior $\l \in \Delta(\V)$, define its Bayes headline and excess posterior risk as
\begin{align*}
    a(\l|Y) = \frac{\sum_{k = 1}^{K} \l_{k}Y_{k}/\sigma_{k}^{2}}{\sum_{k = 1}^{K} \l_{k}/\sigma_{k}^{2}}, \quad \rho(\l|Y) = \sum_{k = 1}^{K} \l_{k}D_{k}(a(\l|Y)|Y)^{2} = \frac{\sum_{j < k} (Y_{j} - Y_{k})^{2}(\l_{j}/\sigma_{j}^{2})(\l_{k}/\sigma_{k}^{2})}{\sum_{\ell = 1}^{K} \l_{\ell}/\sigma_{\ell}^{2}}.
\end{align*}
The following result characterizes the headline estimation procedure.

\begin{proposition}[Least favorable design priors]\label{arxiv1:prop:least.favorable}
The optimal headline $\htheta_{\L}^{*}(Y)$ exists uniquely. Moreover, there exists $\l^{*} \in \L$ such that
\begin{align*}
    \l^{*} \in \arg\max_{\l \in \L} \rho(\l|Y), \quad \htheta_{\L}^{*}(Y) = a(\l^{*}|Y), \quad R_{\L}^{*}(Y) = \rho(\l^{*}|Y) + 1.
\end{align*}
That is, the optimal headline can be represented as an inverse-variance-weighted average under a least favorable design prior $\l^{*} \in \L$. Moreover, $\l^{*}$ is a worst-case design prior at $a = \htheta_{\L}^{*}(Y)$:
\begin{align*}
    R_{\L}(\htheta_{\L}^{*}(Y)|Y) = R_{\l^{*}}(\htheta_{\L}^{*}(Y)|Y).
\end{align*}
For $K = 2$, define $p_{\Delta} = \sigma_{1}/(\sigma_{1} + \sigma_{2})$, which satisfies (uniquely when $Y_{1} \neq Y_{2}$)
\begin{align*}
    p_{\Delta} \in \arg\max_{p \in [0,1]} \rho(\l(p)|Y), \quad \rho(\l(p)|Y) = \frac{p(1 - p)(Y_{1} - Y_{2})^{2}}{p\sigma_{2}^{2} + (1 - p)\sigma_{1}^{2}}, \quad \l(p) = (p, 1 - p)'.
\end{align*}
Define the interval of $\L$ probabilities on design 1 and the projection of $p_{\Delta}$ onto this interval as
\begin{align*}
    \mathcal{I}_{\L} = \curly{p \in [0,1]: \l(p) \in \L} = \brack{\underline{p}_{\L}, \overline{p}_{\L}}, \quad p_{\L}^{*} = \arg\min_{p \in \mathcal{I}_{\L}}\abs{p - p_{\Delta}} = \min\curly{\max\curly{p_{\Delta}, \underline{p}_{\L}}, \overline{p}_{\L}}.
\end{align*}
In this case, the optimal headline and optimized risk are
\begin{align*}
    \htheta_{\L}^{*}(Y) = \frac{p_{\L}^{*}Y_{1}/\sigma_{1}^{2} + (1 - p_{\L}^{*})Y_{2}/\sigma_{2}^{2}}{p_{\L}^{*}/\sigma_{1}^{2} + (1 - p_{\L}^{*})/\sigma_{2}^{2}}, \quad R_{\L}^{*}(Y) = \frac{p_{\L}^{*}(1 - p_{\L}^{*})(Y_{1} - Y_{2})^{2}}{p_{\L}^{*}\sigma_{2}^{2} + (1 - p_{\L}^{*})\sigma_{1}^{2}} + 1.
\end{align*}
If $Y_{1} \neq Y_{2}$, then $\l(p_{\L}^{*})$ is the unique least favorable design prior. If $Y_{1} = Y_{2}$, then every $\l \in \L$ is least favorable and $\htheta_{\L}^{*}(Y) = Y_{1} = Y_{2}$ with $R_{\L}^{*}(Y) = 1$.
\end{proposition}

\begin{proof}
See Appendix \ref{arxiv1:app:proof:least.favorable}.      
\end{proof}

\begin{remark}[Computation]
The function $\l \mapsto \rho(\l|Y)$ is continuous and concave, where concavity follows because $\rho(\l|Y) = \min_{a \in \R}\sum_{k = 1}^{K}\l_{k}D_{k}(a|Y)^{2}$ is the pointwise minimum of functions that are affine in $\l$. Consequently, computing a least favorable prior amounts to maximizing a concave objective over the compact convex set $\L$, which yields a standard finite-dimensional convex-optimization problem when $\L$ is represented by convex constraints. Once any maximizing $\l^{*}$ is obtained, the optimal headline and optimized risk follow directly from $a(\l^{*}|Y)$ and $\rho(\l^{*}|Y) + 1$. The least favorable prior need not be unique, but every maximizing prior yields the same unique optimal headline and the same optimized risk. 
\end{remark}

\begin{remark}[Equivariance under affine transformations]
The optimal headline is equivariant under affine transformations in the sense that using $(\widetilde{Y}_{k}, \widetilde{\sigma}_{k}) = (\xi_{1} + \xi_{2}Y_{k}, |\xi_{2}|\sigma_{k})$ with $\xi_{1} \in \R$ and $\xi_{2} \neq 0$ yields headline $\xi_{1} + \xi_{2}\htheta_{\L}^{*}(Y)$. The optimized risk is unchanged. Thus, a reader interested in the affine transformation $\xi_{1} + \xi_{2}\theta$ can obtain the corresponding optimal headline by applying the same transformation to the given headline $\htheta_{\L}^{*}(Y)$. This shortcut does not extend to a nonlinear transformation $h(\theta)$, since $h(\htheta_{\L}^{*}(Y))$ generally differs from the optimal headline obtained by formulating the decision problem for $h(\theta)$ on the transformed scale. 
\end{remark}

\begin{remark}[Off-diagonal covariance] 
The optimal headline and optimized risk depend on $\Sigma$ only through the design variances $\sigma_{1}^{2}, \ldots, \sigma_{K}^{2}$. This is a general implication of squared error loss under $\mu$-posterior $\pi_{0}(\cdot|Y) = N(Y,\Sigma)$. Off-diagonal covariance terms affect the sampling distribution of $(\htheta_{\L}^{*}(Y), R_{\L}^{*}(Y))$, but they do not enter the decision problem under the default specifications.
\end{remark}

\subsection{Unrestricted-Ambiguity Benchmark}\label{arxiv1:sec:characterization.full.ambiguity}
Under the simplex class $\L_{\Delta} = \Delta(\V)$, for any candidate headline $a$, a worst-case design prior can be chosen to place probability one on a design that attains the largest standardized discrepancy $\max_{k \in [K]}D_{k}(a|Y)$. Therefore, the optimal headline and optimized risk reduce to
\begin{align*}
    \htheta_{\Delta}^{*}(Y) = \arg\min_{a \in \R} \max_{k \in [K]} D_{k}(a|Y), \quad R_{\Delta}^{*}(Y) = \paren{\min_{a \in \R}\max_{k \in [K]}D_{k}(a|Y)}^{2} + 1.
\end{align*}
Define the \textit{reconciliation cost} $c_{\Delta}^{*}(Y) = \sqrt{R_{\Delta}^{*}(Y) - 1}$.

\begin{proposition}[Unrestricted ambiguity]\label{arxiv1:prop:unrestricted.simplex}
The optimal headline and optimized risk are
\begin{align*}
    \htheta_{\Delta}^{*}(Y) = \max_{k \in [K]}\curly{Y_{k} - \sigma_{k}c_{\Delta}^{*}(Y)} = \min_{k \in [K]}\curly{Y_{k} + \sigma_{k}c_{\Delta}^{*}(Y)}, \quad R_{\Delta}^{*}(Y)
    = \paren{\max_{j < k}\frac{\abs{Y_{j} - Y_{k}}}{\sigma_{j} + \sigma_{k}}}^{2} + 1.
\end{align*}
If $(j^{*},k^{*})$ attains the pairwise maximum in $R_{\Delta}^{*}(Y)$, then
\begin{align*}
    \htheta_{\Delta}^{*}(Y) = \frac{\sigma_{k^{*}}Y_{j^{*}} + \sigma_{j^{*}}Y_{k^{*}}}{\sigma_{j^{*}} + \sigma_{k^{*}}}, \quad c_{\Delta}^{*}(Y) = D_{j^{*}}(\htheta_{\Delta}^{*}(Y)|Y) = D_{k^{*}}(\htheta_{\Delta}^{*}(Y)|Y) = \frac{\abs{Y_{j^{*}} - Y_{k^{*}}}}{\sigma_{j^{*}} + \sigma_{k^{*}}}.
\end{align*}
The reconciliation cost $c_{\Delta}^{*}(Y)$ is the smallest standard-error multiplier $c \geq 0$ such that the design intervals $I_{k}(c) = [Y_{k} - c\sigma_{k}, Y_{k} + c\sigma_{k}]$ have nonempty intersection. In particular,
\begin{align*}
    \bigcap_{k = 1}^{K} I_{k}(c_{\Delta}^{*}(Y)) = \curly{\htheta_{\Delta}^{*}(Y)}, \quad c_{\Delta}^{*}(Y) = \min\curly{c \geq 0: \bigcap_{k = 1}^{K} I_{k}(c) \neq \varnothing}.
\end{align*}
For $K = 2$, the solution is given by the above formulas with $(j^{*}, k^{*}) = (1,2)$. If $Y_{1} \neq Y_{2}$, the unique least favorable prior is $\l^{*} = (p_{\Delta}, 1 - p_{\Delta})'$, where $p_{\Delta} = \sigma_{1}/(\sigma_{1} + \sigma_{2})$ as in Proposition \ref{arxiv1:prop:least.favorable}.
\end{proposition}

\begin{proof}
See Appendix \ref{arxiv1:app:proof:unrestricted.simplex}.        
\end{proof}

For $K = 2$, the solution is
\begin{align*}
    \htheta_{\Delta}^{*}(Y)
    = \frac{\sigma_{1}^{-1}}{\sigma_{1}^{-1} + \sigma_{2}^{-1}}Y_{1} + \frac{\sigma_{2}^{-1}}{\sigma_{1}^{-1} + \sigma_{2}^{-1}}Y_{2}, \quad
    R_{\Delta}^{*}(Y) = \paren{\frac{\abs{Y_{1} - Y_{2}}}{\sigma_{1} + \sigma_{2}}}^{2} + 1.
\end{align*}
In this case, the optimal headline $ \htheta_{\Delta}^{*}(Y)$ is an inverse-standard-deviation-weighted average of the two design estimates. It places higher weight on the more precise design, but less aggressively than the conventional inverse-variance-weighted average used when estimates are independent and unbiased for a common parameter. The latter procedure chooses weights to minimize the variance of a weighted estimator, while the above headline estimation procedure equalizes the two standardized discrepancies to account for ambiguity over research designs. The excess risk $ R_{\Delta}^{*}(Y)-1$ is the squared standardized separation between the two estimates.

For $K > 2$, an analogous solution is determined by the active pair $(Y_{j^{*}},Y_{k^{*}})$ attaining the largest standardized separation. The remaining estimates do not affect the headline unless they change which pair is active. A risk of one means exact numerical agreement across designs. A risk no greater than two means that one-standard-error intervals around each design estimate share a common point. Under the simplex class, the reconciliation cost $c_{\Delta}^{*}(Y)$ is the smallest common standard-error multiplier such that the design intervals share a common point, and $\htheta_{\Delta}^{*}(Y)$ is the unique intersection point at that multiplier.

\begin{remark}[Sensitivity to outliers]
The unrestricted-ambiguity benchmark has the feature of being sensitive to outlier design estimates. Because any one design may be target-valid with arbitrarily high probability under the simplex class, the simplex-optimal headline has to guard against dogmatic design priors. This means that a cluster of proximate design estimates does not necessarily outweigh a single discrepant estimate.
\end{remark}

\begin{remark}[Calibrating the reconciliation cost]
\label{arxiv1:remark:calibrating.cost}
Under the default posterior,
\begin{align*}
    \theta|Y,v = v_{k} \sim N(Y_{k},\sigma_{k}^{2}), \quad \forall k \in [K].
\end{align*}
For $\gamma \in (0,1)$, define the design-specific posterior credible
interval as
\begin{align*}
    C_{k}(\gamma|Y) = [Y_{k}-z_{1-\gamma/2}\sigma_{k}, Y_{k}+z_{1-\gamma/2}\sigma_{k}], \quad z_{1-\gamma/2} = \Phi^{-1}\paren{1-\frac{\gamma}{2}}.
\end{align*}
Conditional on $v=v_{k}$, the credible
interval $C_{k}(\gamma|Y)$ contains $\theta$ with posterior probability $1-\gamma$:
\begin{align*}
    \P[]{\theta \in C_{k}(\gamma|Y) \mid Y,v = v_{k}} = 1 - \gamma.
\end{align*}
This equality continues to hold for data-dependent $\gamma(Y)$,
since the posterior conditions on $Y$. Under the simplex class, the interval characterization in Proposition
\ref{arxiv1:prop:unrestricted.simplex} implies
\begin{align*}
    \gamma_{\Delta}^{*}(Y) = 2\paren{1-\Phi(c_{\Delta}^{*}(Y))} = \sup \curly{\gamma \in (0,1): \bigcap_{k = 1}^{K} C_{k}(\gamma|Y) \neq \varnothing}.
\end{align*}
Thus, $\gamma_{\Delta}^{*}(Y)$ is the largest common posterior
noncoverage level at which the design-specific credible intervals still
share a point. Equivalently, $1-\gamma_{\Delta}^{*}(Y)$ is the minimum
common credibility level required for mutual overlap. At this boundary, the
first common point is $\htheta_{\Delta}^{*}(Y)$.

The $95\%$ credibility level benchmark is given by $c_{\Delta}^{*}(Y) = z_{0.975} \approx 1.96$. For ease of reporting, I use the rounded $c = 2$ as a cost benchmark, which corresponds to risk $R = 2^{2} + 1 = 5$---a 400\% increase over the oracle risk. Under the simplex, a cost above two therefore means that even the approximately $95\%$ design-specific credible intervals do not share a common point.
\end{remark}

\begin{remark}[Reconciliation cost for other classes]
\label{arxiv1:remark:cost.other.classes}
More generally for the other classes $\L$ under consideration, define the reconciliation cost
\begin{align*}
    c_{\L}^{*}(Y) = \sqrt{R_{\L}^{*}(Y)-1}.
\end{align*}
For $\L \neq \L_{\Delta}$, the reconciliation cost does not generally have an interval-overlap interpretation. I nevertheless use $c = 2$, or equivalently $R = 5$, as a \textit{simplex-calibrated} benchmark for the other classes of interest.\footnote{$c_{\L}^{*}(Y) \leq c_{\Delta}^{*}(Y)$ because $\L \subseteq \Delta(\V) = \L_{\Delta}$. Thus, $c_{\L}^{*}(Y) > 2$ implies that $c_{\Delta}^{*}(Y)$ exceeds the two-standard-error multiplier benchmark. However, $c_{\L}^{*}(Y) \leq 2$ need not imply that all $95\%$ design-specific credible intervals overlap.}
\end{remark}

\subsection{Structured Ambiguity Classes}\label{arxiv1:sec:characterization.structured.ambiguity}
The simplex class accommodates cases where the researcher does not want to rule out any possible design prior. In other cases, the researcher's ambiguity may have additional structure, as with the perturbation and ratio classes defined in \eqref{arxiv1:eq:perturbation.class} and \eqref{arxiv1:eq:ratio.class}. I now characterize these structured ambiguity classes. They can serve as primary specifications when their restrictions are substantively justified or as sensitivity analyses relative to the unrestricted-ambiguity benchmark.

\subsubsection{Baseline-Prior Perturbation Class}
Consider the perturbation class $\L_{\epsilon}(\l^{0}) = (1 - \epsilon)\l^{0} + \epsilon \Delta(\V)$ for a baseline design prior $\l^{0} \in \Delta(\V)$.

\begin{proposition}[Baseline-prior ambiguity]\label{arxiv1:prop:epsilon.perturbation}
For perturbation level $\epsilon \in [0,1]$,
\begin{align*}
    \htheta_{\epsilon}^{*}(Y) = \arg\min_{a \in \R} R_{\epsilon}(a|Y), \quad R_{\epsilon}(a|Y) = \paren{(1 - \epsilon)\sum_{k = 1}^{K} \l_{k}^{0}D_{k}(a|Y)^{2} + \epsilon\max_{k \in [K]}D_{k}(a|Y)^{2}} + 1.
\end{align*}
Let $R_{\epsilon}^{*}(Y) = R_{\epsilon}(\htheta_{\epsilon}^{*}(Y)|Y)$. By Proposition \ref{arxiv1:prop:least.favorable}, the solution can be computed as
\begin{align*}
    \l_{\epsilon}^{*} \in \arg\max\curly{\rho(\l|Y): \l \in \Delta(\V), \l \geq (1-\epsilon)\l^{0}}, \quad \htheta_{\epsilon}^{*}(Y) = a(\l_{\epsilon}^{*}|Y), \quad
    R_{\epsilon}^{*}(Y) = \rho(\l_{\epsilon}^{*}|Y) + 1.
\end{align*}
The optimal headline is continuous in $\epsilon$ and follows a monotone path between its endpoints: if $\htheta_{0}^{*}(Y) < \htheta_{1}^{*}(Y)$, then $\epsilon \mapsto \htheta_{\epsilon}^{*}(Y)$ is nondecreasing; if $\htheta_{0}^{*}(Y) > \htheta_{1}^{*}(Y)$, then it is nonincreasing; and if the endpoints coincide, then the path is constant. At the endpoints,
\begin{align*}
    \htheta_{0}^{*}(Y) = a(\l^{0}|Y) = \frac{\sum_{k = 1}^{K}\l_{k}^{0}Y_{k}/\sigma_{k}^{2}}{\sum_{k = 1}^{K}\l_{k}^{0}/\sigma_{k}^{2}}, \qquad \htheta_{1}^{*}(Y) = \htheta_{\Delta}^{*}(Y).
\end{align*}
For $K = 2$, $p^{0} = \P[\l^{0}]{v = v_{1}}$, and $p_{\Delta} = \sigma_{1}/(\sigma_{1} + \sigma_{2})$, Proposition \ref{arxiv1:prop:least.favorable} gives
\begin{align*}
    p_{\epsilon}^{*} = \arg\min_{p \in \mathcal{I}_{\epsilon}}\abs{p - p_{\Delta}} = \min\curly{\max \curly{p_{\Delta}, (1 - \epsilon)p^{0}}, (1 - \epsilon)p^{0} + \epsilon}, \quad \mathcal{I}_{\epsilon} = [(1 - \epsilon)p^{0},(1 - \epsilon)p^{0} + \epsilon].
\end{align*}
In this case, the optimal headline and optimized risk are
\begin{align}\label{arxiv1:eq:perturbation.solution.K=2}
    \htheta_{\epsilon}^{*}(Y) = \frac{p_{\epsilon}^{*}Y_{1}/\sigma_{1}^{2} + (1 - p_{\epsilon}^{*})Y_{2}/\sigma_{2}^{2}}{p_{\epsilon}^{*}/\sigma_{1}^{2} + (1 - p_{\epsilon}^{*})/\sigma_{2}^{2}}, \quad R_{\epsilon}^{*}(Y)
    = \frac{p_{\epsilon}^{*}(1 - p_{\epsilon}^{*})(Y_{1} - Y_{2})^{2}}{p_{\epsilon}^{*}\sigma_{2}^{2} + (1 - p_{\epsilon}^{*})\sigma_{1}^{2}} + 1.
\end{align}
If $p_{\Delta} \in \mathcal{I}_{\epsilon}$, then $\htheta_{\epsilon}^{*}(Y) = \htheta_{\Delta}^{*}(Y)$. In particular, when $p^{0} \in (0,1)$, the perturbation-optimal headline coincides with the simplex-optimal headline if  $\epsilon \geq \max\{1-(p_{\Delta}/p^{0}), 1-(1-p_{\Delta})/(1-p^{0})\}$.
\end{proposition}

\begin{proof}
See Appendix \ref{arxiv1:app:proof:epsilon.perturbation}.
\end{proof}

The excess risk criterion $R_{\epsilon}(a|Y)-1$ is a convex combination of (i) the baseline design prior average of the squared standardized discrepancies and (ii) the maximum discrepancy, where the perturbation level $\epsilon$ controls how much the headline estimation procedure protects against the baseline design prior being incorrect. At $\epsilon = 0$, the researcher fully commits to the baseline prior, leading to Bayes headline $a(\l^{0}|Y)$. At $\epsilon = 1$, the perturbation class coincides with the simplex class, leading to the unrestricted-ambiguity headline $\htheta_{\Delta}^{*}(Y)$. The perturbation-optimal headline moves monotonically between these two endpoint headlines---or remains constant when they coincide. In particular, when $\l^{0} = \1/K$, the endpoint at $\epsilon = 0$ is the inverse-variance-weighted average; when the design estimates are independent, this is the conventional precision-weighted average. For sensitivity analysis relative to the unrestricted-ambiguity benchmark, a researcher can trace the perturbation-optimal headline and optimized risk over a range of $\epsilon$ or report the smallest perturbation level at which a qualitative conclusion changes. 

For $K = 2$ and an equal-probability baseline $\l^{0} = \1/2 = (1/2, 1/2)'$, solution \eqref{arxiv1:eq:perturbation.solution.K=2} yields
\begin{align*}
    \htheta_{\epsilon}^{*}(Y)
    = \frac{p_{\epsilon}^{*}Y_{1}/\sigma_{1}^{2} + (1 - p_{\epsilon}^{*})Y_{2}/\sigma_{2}^{2}}{p_{\epsilon}^{*}/\sigma_{1}^{2} + (1 - p_{\epsilon}^{*})/\sigma_{2}^{2}}, \quad p_{\epsilon}^{*} = \min\curly{\max \curly{p_{\Delta}, (1 - \epsilon)\frac{1}{2}}, (1 - \epsilon)\frac{1}{2} + \epsilon}, \quad \epsilon \in [0,1].
\end{align*}
The corresponding endpoints are
\begin{align*}
    \htheta_{0}^{*}(Y) = a(\l^{0}|Y) = \frac{Y_{1}/\sigma_{1}^{2} + Y_{2}/\sigma_{2}^{2}}{1/\sigma_{1}^{2} + 1/\sigma_{2}^{2}}, \quad \htheta_{1}^{*}(Y) = \htheta_{\Delta}^{*}(Y) = \frac{Y_{1}/\sigma_{1} + Y_{2}/\sigma_{2}}{1/\sigma_{1} + 1/\sigma_{2}}.
\end{align*}
Thus, the perturbation-optimal headline moves from the inverse-variance-weighted average toward the inverse-standard-deviation-weighted average. Every point on this path lies between those two endpoint headlines and can therefore be written as a convex combination of them, although the convex-combination weight is generally not equal to $\epsilon$.

\subsubsection{Favored-Design Ratio Class}
Consider the ratio class $\L_{r} = \{\l: \l_{1} \geq r\l_{k},\ \forall k \neq 1\}$. Denote a decreasing rearrangement of the squared standardized discrepancies $\curly{D_{k}(a|Y)^{2}: k \neq 1}$ as $D_{-1,(1)}(a|Y)^{2} \geq \cdots \geq D_{-1,(K - 1)}(a|Y)^{2}$.

\begin{proposition}[Favored-design ambiguity]\label{arxiv1:prop:preferred.design}
For ratio parameter $r \geq 1$,
\begin{align*}
    \htheta_{r}^{*}(Y) = \arg\min_{a \in \R} R_{r}(a|Y), \quad R_{r}(a|Y) = \max_{m \in \{0,\ldots,K - 1\}}\frac{rD_{1}(a|Y)^{2} + \sum_{j = 1}^{m}D_{-1,(j)}(a|Y)^{2}}{r + m} + 1.
\end{align*}
For each $m$, the corresponding term in the maximum is generated by an extreme design prior that places probability $r/(r+m)$ on the favored design and probability $1/(r+m)$ on each of the $m$ alternatives with the largest squared standardized discrepancies. If $m = 0$, the inner sum is set to zero. Let $R_{r}^{*}(Y) = R_{r}(\htheta_{r}^{*}(Y)|Y)$. By Proposition \ref{arxiv1:prop:least.favorable}, the solution can be computed as
\begin{align*}
    \l_{r}^{*} \in \arg\max\curly{\rho(\l|Y): \l \in \Delta(\V), \l_{1} \geq r\l_{k}, \forall k \neq 1}, \quad \htheta_{r}^{*}(Y) = a(\l_{r}^{*}|Y), \quad R_{r}^{*}(Y) = \rho(\l_{r}^{*}|Y) + 1.
\end{align*}
The optimized risk $R_{r}^{*}(Y)$ is nonincreasing in $r$ and, as $r \to \infty$,
\begin{align*}
    R_{r}^{*}(Y) \to 1, \quad \htheta_{r}^{*}(Y) \to Y_{1}.
\end{align*}
For $K = 2$, $p_{r} = r/(r+1)$, and $p_{\Delta} = \sigma_{1}/(\sigma_{1} + \sigma_{2})$, Proposition \ref{arxiv1:prop:least.favorable} gives
\begin{align*}
    p_{r}^{*} = \arg\min_{p \in \mathcal{I}_{r}}\abs{p - p_{\Delta}} = \max\curly{p_{\Delta}, p_{r}}, \quad \mathcal{I}_{r} = [p_{r}, 1].
\end{align*}
In this case, the optimal headline and optimized risk are
\begin{align*}
    \htheta_{r}^{*}(Y) = \frac{p_{r}^{*}Y_{1}/\sigma_{1}^{2} + (1 - p_{r}^{*})Y_{2}/\sigma_{2}^{2}}{p_{r}^{*}/\sigma_{1}^{2} + (1 - p_{r}^{*})/\sigma_{2}^{2}}, \quad R_{r}^{*}(Y) = \frac{p_{r}^{*}(1 - p_{r}^{*})(Y_{1} - Y_{2})^{2}}{p_{r}^{*}\sigma_{2}^{2} + (1 - p_{r}^{*})\sigma_{1}^{2}} + 1.
\end{align*}
If $p_{\Delta} \in \mathcal{I}_{r}$, then $\htheta_{r}^{*}(Y) = \htheta_{\Delta}^{*}(Y)$. Equivalently, the ratio-optimal and simplex-optimal headlines coincide if $r \leq \sigma_{1}/\sigma_{2}$.
\end{proposition}

\begin{proof}
See Appendix \ref{arxiv1:app:proof:preferred.design}.
\end{proof}

The extreme-prior formula gives a concrete interpretation of the criterion $R_{r}(a|Y)$. Think of the favored design as receiving $r$ units of prior probability and each selected alternative as receiving one unit. For a candidate headline and a fixed number $m$ of alternatives, a worst-case extreme prior selects the $m$ designs with the largest squared standardized discrepancies. The criterion then chooses the value of $m$ that gives the largest weighted average discrepancy. If the maximizing extreme prior is unique at the optimal headline, it is least favorable. When several extreme priors tie, the least-favorable prior may instead be a mixture, as described in Appendix \ref{arxiv1:app:proof:preferred.design}.

The class $\L_{r}$ shrinks as $r$ increases. Consequently, the optimized risk shrinks toward one and the optimal headline approaches the favored-design estimate. The ratio class lets a researcher quantify how strongly a favored design must be privileged before the optimal headline changes materially. A useful implementation reports the optimal headline and risk over a grid of $r$. This allows one to see where $\htheta_{r}^{*}(Y)$ changes sign or crosses a policy threshold, or to examine where $R_{r}^{*}(Y)$ is low relative to some benchmark---such as the simplex-optimized risk $R_{\Delta}^{*}(Y)$.

\section{Asymptotic Results}\label{arxiv1:sec:asymptotic.results}
The foregoing results are developed under normally distributed estimates $Y \sim N(\mu, \Sigma)$ and a known covariance matrix $\Sigma$. In practice, the researcher can implement the headline estimation procedure with non-normal estimates and an estimated covariance matrix. I study this plug-in procedure below, which I propose as the default practical implementation.

For sample size $n \to \infty$, let $Y_{n} = (Y_{1,n}, \ldots, Y_{K,n})'$ denote the vector of design estimates and let $\tSigma_{n}$ denote an estimated positive definite covariance matrix for $Y_{n}$. The corresponding design standard errors and standardized discrepancies are
\begin{align*}
    \tsigma_{k,n} = \sqrt{v_{k}'\tSigma_{n}v_{k}}, \quad D_{k,n}(a|Y_{n}) = \frac{\abs{Y_{k,n} - a}}{\tsigma_{k,n}}.
\end{align*}
Fix a nonempty, compact, and convex class of design priors $\L \subseteq \Delta(\V)$. The proposed plug-in implementation simply replaces $(Y, \Sigma)$ in Section \ref{arxiv1:sec:proposed.procedures} with $(Y_{n}, \tSigma_{n})$, which yields
\begin{align*}
    \htheta_{n,\L}^{*} = \arg\min_{a \in \R}\max_{\l \in \L}\sum_{k = 1}^{K}\l_{k}D_{k,n}(a|Y_{n})^{2}, \quad R_{n,\L}^{*} = \min_{a \in \R}\max_{\l \in \L}\sum_{k = 1}^{K}\l_{k}D_{k,n}(a|Y_{n})^{2} + 1.
\end{align*}
In words, the plug-in optimal headline and plug-in risk $(\htheta_{n,\L}^{*}, R_{n,\L}^{*})$ are obtained by applying the headline estimation criterion directly to the estimates $Y_{k,n}$ and standard errors $\tsigma_{k,n}$. The formulas developed previously for classes $(\L_{\Delta}, \L_{\epsilon}(\l^{0}), \L_{r})$ therefore apply directly.

\begin{assumption}[Local design disagreement]\label{arxiv1:ass:local.gaussian}
For target-valid design $k_{0} \in [K]$, parameter $\theta \in \R$, and local discrepancy vector $b = (b_{1}, \ldots, b_{K})' \in \R^{K}$ satisfying $b_{k_{0}} = 0$,
\begin{align*}
    \sqrt{n}(Y_{n} - \mu_{n}) \to[d] N(0,\Sigma), \quad \mu_{n} = \theta\1 + \frac{b}{\sqrt{n}},
\end{align*}
where $\Sigma$ is positive definite and $n\tSigma_{n} \to[p] \Sigma$.
\end{assumption}

Assumption \ref{arxiv1:ass:local.gaussian} considers differences in estimands that are on the same order as sampling uncertainty. The target-valid design estimand satisfies $\mu_{k_{0},n} = \theta$, while the other estimands may differ from $\theta$ by $b_{k}/\sqrt{n}$. The purpose of this local sequence is to produce a nondegenerate limit decision problem. Define objects for the Gaussian limit decision problem as
\begin{align*}
    \htheta_{\L}^{*}(z) = \arg\min_{a \in \R}\max_{\l \in \L}\sum_{k = 1}^{K}\l_{k}\frac{(z_{k} - a)^{2}}{\sigma_{k}^{2}}, \quad R_{\L}^{*}(z)
    = \min_{a \in \R}\max_{\l \in \L}\sum_{k = 1}^{K}\l_{k}\frac{(z_{k} - a)^{2}}{\sigma_{k}^{2}} + 1, \quad z \in \R^{K}.
\end{align*} 

\begin{proposition}[Local asymptotic behavior of the plug-in procedure]\label{arxiv1:prop:plugin.local}
Under Assumption \ref{arxiv1:ass:local.gaussian},
\begin{align*}
    \sqrt{n}(\htheta_{n,\L}^{*} - \theta) \to[d] \htheta_{\L}^{*}(Z), \quad
    R_{n,\L}^{*} \to[d] R_{\L}^{*}(Z), \quad Z \sim N(b,\Sigma).
\end{align*}
\end{proposition}

\begin{proof}
See Appendix \ref{arxiv1:app:proof:plugin.local}.
\end{proof}

Proposition \ref{arxiv1:prop:plugin.local} gives the first-order asymptotic behavior of the plug-in headline estimation procedure. The normalized headline has the same limit as the optimal headline in the Gaussian limit decision problem, while the plug-in risk converges to the corresponding limit risk.

I next consider asymptotics with persistent design disagreement, where differences among design estimands do not shrink with sampling uncertainty.

\begin{proposition}[Persistent design disagreement]\label{arxiv1:prop:persistent.disagreement}
Suppose $Y_{n} \to[p] \mu^{0}$ and $n\tSigma_{n} \to[p] \Sigma$, where $\Sigma$ is positive definite. Define
\begin{align*}
    \theta_{\L}^{\mathrm{lim}}(\mu^{0}) = \arg\min_{a \in \R}\max_{\l \in \L}\sum_{k = 1}^{K}\l_{k}\frac{(\mu_{k}^{0} - a)^{2}}{\sigma_{k}^{2}}, \quad R_{\L}^{\mathrm{lim}}(\mu^{0}) = \min_{a \in \R}\max_{\l \in \L}\sum_{k = 1}^{K}\l_{k}\frac{(\mu_{k}^{0} - a)^{2}}{\sigma_{k}^{2}}.
\end{align*}
Then
\begin{align*}
    \htheta_{n,\L}^{*} \to[p] \theta_{\L}^{\mathrm{lim}}(\mu^{0}), \quad \frac{R_{n,\L}^{*}}{n} \to[p] R_{\L}^{\mathrm{lim}}(\mu^{0}).
\end{align*}
If every design receives positive probability under at least one prior in $\L$ and the coordinates of $\mu^{0}$ are not all equal, then $R_{\L}^{\mathrm{lim}}(\mu^{0}) > 0$ so that $R_{n,\L}^{*} \to[p] \infty$.
\end{proposition}

\begin{proof}
See Appendix \ref{arxiv1:app:proof:persistent.disagreement}.
\end{proof}

Taken together, Propositions \ref{arxiv1:prop:plugin.local} and \ref{arxiv1:prop:persistent.disagreement} characterize the plug-in procedure in two asymptotic regimes. Under local disagreement, the differences among design estimands remain on the same order as sampling uncertainty and the normalized plug-in headline converges in distribution to the optimal headline for the corresponding Gaussian limit problem. Under persistent disagreement with $R_{\L}^{\mathrm{lim}}(\mu^{0}) > 0$, the headline converges in probability to a deterministic compromise that need not equal the target, while the diverging risk increasingly makes visible the difficulty of reconciling a common headline across designs.

\section{Empirical Applications}\label{arxiv1:sec:empirical.applications}
This section illustrates the headline estimation procedure in empirical applications to three papers: \citet{chen2026women}, \citet{garin2025impact}, and \citet{ganong2024spending}. For each application, I define the target parameter, candidate designs, baseline headline $a_{B}$, and class of design priors $\L$. The outputs of interest include the optimal headline $\htheta_{\L}^{*}$, baseline risk $R_{\L}^{B}$, optimized risk $R_{\L}^{*}$, active designs or least favorable design priors, and percentage reduction in risk when moving from the baseline to the optimal headline:
\begin{align*}
    \paren{1 - \frac{R_{\L}^{*}}{R_{\L}^{B}}} \times 100.
\end{align*}
The oracle risk is one, so a risk value of $R$ means a $(R - 1) \times 100$ percentage increase above the oracle value. When interpreting magnitudes, I also report the reconciliation cost $c_{\L}^{*} = \sqrt{R_{\L}^{*} - 1}$. Following Remarks \ref{arxiv1:remark:calibrating.cost} and \ref{arxiv1:remark:cost.other.classes}, I use $c = 2$ and $R = 5$ as simplex-calibrated benchmarks for cost and risk. Under the simplex class, exceeding this benchmark means that the $95\%$ design-specific credible intervals fail to share a common point. Although this interval-overlap characterization does not hold more generally for other classes, the cost and risk calibrations still serve as useful benchmarks. The calculations below are based on point estimates and standard errors from the papers' published tables. Because these inputs are rounded, the resulting outputs may differ slightly from calculations based on the papers' underlying data. 

\subsection{Returning to the Courtroom-Broadcasting Example}
Section \ref{arxiv1:sec:example.setting} introduced the courtroom-broadcasting setting of \citet{chen2026women} and gave results for their all-litigants sample and the simplex class $\L_{\Delta}$. I now return to that example and additionally consider the two-litigant sample and ratio classes $\L_{r}$ that privilege the paper's primary DiD design over their alternative Bartik IV design.

\begin{itemize}
    \item \textit{Target.} The effect of broadcasting intensity on the female-male plaintiff win-rate gap.
    \item \textit{Candidate designs.} The DiD and Bartik IV designs.\footnote{\citet[page 1576]{chen2026women} say that ``While the Bartik IV and DID strategies rely on distinct identifying assumptions, our results remain remarkably consistent across both approaches.''}
    \item \textit{Baseline headline.} The DiD estimate.
    \item \textit{Design-prior class.} The simplex class $\L_{\Delta}$, followed by the favored-DiD ratio class $\L_{r}$.
\end{itemize}

Table \ref{arxiv1:tab:chen.application} reports the simplex results using the published estimates in \citet[Tables 2 and 3]{chen2026women}. The all-litigants row reproduces the calculations from Section \ref{arxiv1:sec:example.setting}. For the two-litigant sample, the optimal headline is $\htheta_{\Delta}^{*} = 0.0651$, so a 10 percentage-point increase in broadcasting intensity corresponds to a 0.651 percentage-point narrowing of the gender gap.

\begin{table}[H]
\centering
\caption{Simplex-optimal headline estimation: \citet{chen2026women}}
\label{arxiv1:tab:chen.application}
\vspace{0.25em}
\makebox[\textwidth][c]{%
\begin{minipage}{1.05\textwidth}
\begin{threeparttable}
\small
\setlength{\tabcolsep}{2.5pt}
\renewcommand{\arraystretch}{1.12}
\begin{tabularx}{\linewidth}{@{}>{\raggedright\arraybackslash}p{0.14\linewidth}*{6}{>{\centering\arraybackslash}X}@{}}
\toprule
Sample & DiD estimate (SE) & Bartik IV estimate (SE) & Optimal headline $\htheta_{\Delta}^{*}$ & Baseline risk $R_{\Delta}^{B}$ & Optimal risk $R_{\Delta}^{*}$ & Risk reduction \\
\midrule
All litigants & 0.0394 (0.00255) & 0.0726 (0.00599) & 0.0493 & 31.72 & 16.11 & 49.2\% \\
Two litigants & 0.0536 (0.00336) & 0.0889 (0.00699) & 0.0651 & 26.50 & 12.63 & 52.3\% \\
\bottomrule
\end{tabularx}
\begin{tablenotes}[flushleft]
\footnotesize
\item \textit{Notes}. The estimates and standard errors are reported in \citet[Tables 2 and 3]{chen2026women}. The DiD estimate serves as the baseline headline $a_{B}$. Risk reduction is $(1 - (R_{\Delta}^{*}/R_{\Delta}^{B})) \times 100$.
\end{tablenotes}
\end{threeparttable}
\end{minipage}}
\end{table}

The optimal headlines reduce the risk corresponding to baseline DiD headlines by roughly half. Even after optimization, however, risk remains high because the differences between the DiD and Bartik IV estimates are large relative to their small standard errors. The optimized risks are 16.11 and 12.63 times the oracle risk, corresponding to costs $c_{\Delta}^{*} \approx 3.89$ and $c_{\Delta}^{*} \approx 3.41$. Both exceed the cost benchmark. Equivalently, the approximately $95\%$ design-specific credible intervals fail to share a common point. Thus, while the simplex-optimal headlines preserve the conclusion that broadcasting narrows the gender gap, the DiD and IV designs do not tightly reconcile one exact magnitude under unrestricted design ambiguity.

I now consider the structured ambiguity given by favored-design ratio classes that privilege DiD. In particular, a favored-DiD ratio class $\L_{r}$ requires DiD to be at least $r \geq 1$ times as likely to be target-valid as Bartik IV. Since $\sigma_{\mathrm{DiD}} < \sigma_{\mathrm{IV}}$, Proposition \ref{arxiv1:prop:preferred.design} implies that the least favorable prior places probabilities $r/(r + 1)$ and $1/(r + 1)$ on DiD and IV, respectively. Table \ref{arxiv1:tab:chen.ratio} reports the sensitivity path $r \mapsto (\htheta_{r}^{*}, R_{r}^{*})$ for a range of $r$.

\begin{table}[H]
\centering
\caption{Favored-DiD ratio sensitivity: \citet{chen2026women}}
\label{arxiv1:tab:chen.ratio}
\vspace{0.25em}
\begin{threeparttable}
\footnotesize
\setlength{\tabcolsep}{3.5pt}
\renewcommand{\arraystretch}{1.10}
\begin{tabular}{@{}lcrrrr@{}}
\toprule
& & \multicolumn{2}{c}{All litigants} & \multicolumn{2}{c}{Two litigants} \\
\cmidrule(lr){3-4}\cmidrule(lr){5-6}
Design-prior class & \shortstack{Minimum DiD \\ probability} & Headline $\htheta_{r}^{*}$ & Risk $R_{r}^{*}$ & Headline $\htheta_{r}^{*}$ & Risk $R_{r}^{*}$ \\
\midrule
Unrestricted simplex & -- & 0.0493 & 16.113 & 0.0651 & 12.632 \\
$r = 1$ & 0.500 & 0.0445 & 14.003 & 0.0602 & 11.358 \\
$r = 2$ & 0.667 & 0.0422 & 10.389 & 0.0573 & 8.621 \\
$r = 5$ & 0.833 & 0.0406 & 5.941 & 0.0552 & 5.063 \\
$r = 10$ & 0.909 & 0.0400 & 3.743 & 0.0544 & 3.266 \\
$r = 25$ & 0.962 & 0.0396 & 2.173 & 0.0539 & 1.972 \\
$r = 100$ & 0.990 & 0.0395 & 1.304 & 0.0537 & 1.252 \\
$r \to \infty$ & 1.000 & 0.0394 & 1.000 & 0.0536 & 1.000 \\
\bottomrule
\end{tabular}
\begin{tablenotes}[flushleft]
\footnotesize
\item \textit{Notes}. The favored-DiD ratio class requires $\l_{\mathrm{DiD}} \geq r\l_{\mathrm{IV}}$, so the minimum probability assigned to DiD is $r/(r + 1)$. At every finite $r$, the least favorable design prior places probabilities $r/(r + 1)$ and $1/(r + 1)$ on DiD and IV. The unrestricted row reproduces the simplex class calculations from Table \ref{arxiv1:tab:chen.application}. 
\end{tablenotes}
\end{threeparttable}
\end{table}

At $r=5$, the reconciliation costs are approximately $2.22$ and $2.02$, which remain just above the simplex-calibrated benchmark $c = 2$. At $r=10$, the costs fall to approximately $1.66$ and $1.51$, respectively. Thus, under a favored-DiD ratio of this strength, the reconciliation costs fall below the benchmark and the optimal headlines are close to the DiD estimates.

The $r=10$ class requires the prior probability on DiD to be at least $10/11 \approx 0.909$, so no admissible prior places more than approximately $9.1\%$ probability on IV. This restriction may be substantively reasonable if the Bartik IV design is viewed primarily as a robustness check for the paper's favored DiD design.\footnote{\citet[page 1576]{chen2026women} seem to cast the IV design as a robustness check: ``The primary empirical specification is the generalized difference-in-differences (DID) with continuous treatment, controlling for court-area and year-quarter fixed effects \ldots A potential threat to identification is that the aggregate reform intensity may be endogenous. For instance, it could be contaminated by unobserved confounding factors that vary across court-area and over time. To address this concern, we strengthen our analysis by implementing a Bartik-like instrumental variable (IV) strategy. \ldots The IV estimates are only slightly larger than the corresponding baseline DID estimates, suggesting that potential endogeneity concerns are limited. While the Bartik IV and DID strategies rely on distinct identifying assumptions, our results remain remarkably consistent across both approaches.''} Overall, this sensitivity analysis suggests that, under modest design ambiguity that favors DiD, there is scope for the DiD and Bartik IV designs to tightly reconcile one headline for the effect of broadcasting intensity on the gender win-rate gap.

\subsection{When the Published Headline Is Already Adequate}
\citet{garin2025impact} estimate the long-run effects of a 12-month incarceration sentence on labor-market outcomes and later incarceration exposure using two IV designs based on discontinuities in sentencing guidelines in North Carolina and random judge assignments in Ohio. The authors headline precision-weighted averages across the two states. The corresponding IVs need not identify the same local average treatment effect: the two designs operate in different populations and shift different sentencing margins. However, if one imagines a latent ``effect of a 12-month incarceration sentence for marginal felony defendants,'' then both IV designs plausibly speak to that target parameter, with ambiguity over which design provides the more appropriate link.

\begin{itemize}
    \item \textit{Target.} For each outcome of interest, the effect of a 12-month incarceration sentence for marginal felony defendants.
    \item \textit{Candidate designs.} The North Carolina and Ohio IV designs.\footnote{\citet[page 504]{garin2025impact} note that ``Using multiple research designs allows us to test the sensitivity of our results to empirical strategy.''}
    \item \textit{Baseline headline.} The published precision-weighted average.
    \item \textit{Design-prior class.} The simplex class $\L_{\Delta}$, followed by the perturbation class $\L_{\epsilon}(\l^{0})$ under an equal-probability baseline prior $\l^{0} = (1/2, 1/2)'$.
\end{itemize}

Panel A of Table \ref{arxiv1:tab:garin.application} compares the published precision-weighted averages $a_{B}$ with the simplex-optimal headlines $\htheta_{\Delta}^{*}$. For annual W-2 earnings, the baseline and optimized risks are 1.057 and 1.041; for cumulative earnings, they are 1.376 and 1.262. Adopting the simplex-optimal headline reduces risk by 1.5\% and 8.3\%, respectively. The risks remain close to the oracle benchmark, and the modest improvements show that the published precision-weighted averages are already close to optimal headlines for the target parameter.

Panel B presents a perturbation class sensitivity path $\epsilon \mapsto \htheta_{\epsilon}^{*}$ under the baseline prior that places equal probability on the North Carolina and Ohio IV designs. Treating the two estimates from the distinct state samples as independent, the paper's precision-weighted average has the same inverse-variance form as the uniform-prior headline at $\epsilon = 0$. Calculations from the state-design estimates reproduce the published averages up to small rounding differences. As $\epsilon$ increases, the optimal headline moves toward the inverse-standard-deviation simplex-optimal headline. The column $\Bar{\epsilon}$ reports the smallest perturbation level at which the sensitivity path reaches the simplex-optimal headline. The perturbation class results are consistent with those from the simplex class, so below I focus on the latter.

\begin{table}[H]
\centering
\caption{Simplex-optimal headline estimation and perturbation sensitivity: \citet{garin2025impact}}
\label{arxiv1:tab:garin.application}
\vspace{0.25em}
\makebox[\textwidth][c]{%
\begin{minipage}{1.095\textwidth}
\begin{threeparttable}
\small
\setlength{\tabcolsep}{1.8pt}
\renewcommand{\arraystretch}{1.12}
\begin{tabularx}{0.995\linewidth}{@{}>{\raggedright\arraybackslash}p{0.19\linewidth}*{7}{>{\centering\arraybackslash}X}@{}}
\toprule
\multicolumn{8}{@{}l}{\textit{Panel A. Published and simplex-optimal headlines}} \\
\addlinespace[2pt]
Outcome & N. Carolina (SE) & Ohio (SE) & Published average $a_{B}$ & Baseline risk $R_{\Delta}^{B}$ & Optimal headline $\htheta_{\Delta}^{*}$ & Optimal risk $R_{\Delta}^{*}$ & Risk reduction \\
\midrule
Annual W-2 & \shortstack{113.45\\(223.8)} &
\shortstack{233.97\\(371.5)} & 145.54 & 1.057 & 158.76 & 1.041 & 1.5\% \\
Cumulative W-2 & \shortstack{$-2{,}675$\\(782)} & \shortstack{$-3{,}881$\\(1,576)} & $-2{,}914$ & 1.376 & $-3{,}075$ & 1.262 & 8.3\% \\
Any W-2 & \shortstack{0.024\\(0.010)} & \shortstack{0.004\\(0.013)} & 0.016 & 1.852 & 0.015 & 1.756 & 5.2\% \\
Days incarcerated & \shortstack{3.20\\(3.31)} & \shortstack{13.50\\(2.52)} & 9.72 & 4.880 & 9.05 & 4.121 & 15.5\% \\
\addlinespace[7pt]
\multicolumn{8}{@{}l}{\textit{Panel B. Uniform-prior perturbation path}} \\
\addlinespace[2pt]
Outcome & $\htheta_{0}^{*}$ & $\htheta_{0.10}^{*}$ & $\htheta_{0.25}^{*}$ & $\htheta_{1}^{*}$ & $\Bar{\epsilon}$ & $R_{0}^{*}$ & $R_{1}^{*}$ \\
\midrule
Annual W-2 & 145.54 & 150.48 & 158.76 & 158.76 & 0.248 & 1.039 & 1.041 \\
Cumulative W-2 & $-2{,}913$ & $-2{,}954$ & $-3{,}026$ & $-3{,}075$ & 0.337 & 1.235 & 1.262 \\
Any W-2 & 0.0166 & 0.0156 & 0.0153 & 0.0153 & 0.130 & 1.743 & 1.756 \\
Days incarcerated & 9.72 & 9.23 & 9.05 & 9.05 & 0.136 & 4.065 & 4.121 \\
\bottomrule
\end{tabularx}
\begin{tablenotes}[flushleft]
\footnotesize
\item \textit{Notes}. The estimates, standard errors, and published averages are reported in \citet[Table III]{garin2025impact}. ``Annual W-2'' denotes annual W-2 earnings, ``Cumulative W-2'' denotes cumulative W-2 earnings, and standard errors appear below the corresponding state-design estimates. Annual W-2 earnings, any W-2 earnings, and days incarcerated average outcomes over years 5--9 post-case-filing. Cumulative W-2 earnings are measured as of five years post-case-filing. Panel A evaluates the published average under the simplex class and compares it with the simplex-optimal headline. The risk reduction is $(1 - (R_{\Delta}^{*}/R_{\Delta}^{B})) \times 100$. Panel B uses a uniform prior $\l^{0} = (1/2, 1/2)'$. Here, $\Bar{\epsilon} = \abs{\sigma_{1} - \sigma_{2}}/(\sigma_{1} + \sigma_{2})$ is the smallest perturbation level at which the $\epsilon$-perturbation-optimal $\htheta_{\epsilon}^{*}$ reaches the simplex-optimal $\htheta_{\Delta}^{*}$. Because $R_{\Delta}^{B}$ evaluates the published rounded average while $R_{0}^{*}$ evaluates the recomputed uniform-prior headline, the two risks can differ slightly.
\end{tablenotes}
\end{threeparttable}
\end{minipage}}
\end{table}

The implications of the headline estimation output closely mirror the paper's conclusions. The simplex-optimal headlines for annual W-2 earnings and the probability of having any W-2 earnings are small and positive, reinforcing the finding of no meaningful long-run reduction in earnings or employment. For cumulative W-2 earnings, the headline implies an earnings reduction of about \$3,075, close to the paper's baseline \$2,914, and thus supports the conclusion that earnings lost during incarceration are not subsequently recovered. For long-run incarceration exposure, the headline implies roughly 9 additional days incarcerated per year. Its larger risk indicates weaker agreement about the exact amount of residual incarceration exposure across states, although both designs imply that most of the initial sentence effect has dissipated.

\subsection{A Stable High-MPC Conclusion Across Policy Episodes}
\citet{ganong2024spending} estimate one-month marginal propensities to consume (MPCs) out of unemployment benefits for unemployed U.S. households using six designs based on different policy changes and comparison strategies in the context of the COVID-19 pandemic. A one-month MPC is the change in spending during the first month after a benefit change, divided by the benefit change. The first design compares unemployed households that receive benefits promptly with households that face processing delays. Four designs study the expiration or onset of \$600 and \$300 supplements by comparing unemployed households with matched employed households. The sixth design compares unemployed households in U.S. states that ended the \$300 supplement in June with those in states that ended it in September. The authors report the range of estimated MPCs across designs---0.27 to 0.42---and emphasize that spending responds sharply across all designs. They describe the waiting-for-benefits design as their sharpest identification strategy, while noting that it measures the response to total benefits during the unusually uncertain opening months of the pandemic; the other designs instead isolate supplement changes at other dates. I treat the six designs as candidate links to a latent target one-month MPC out of unemployment benefits.

\begin{itemize}
    \item \textit{Target.} The one-month MPC out of unemployment benefits.
    \item \textit{Candidate designs.} The six designs summarized in \citet[Table 1]{ganong2024spending}.\footnote{\citet[page 2907]{ganong2024spending} say that ``Each of these empirical exercises has distinct advantages and disadvantages, but they all lead to the same conclusion.''}
    \item \textit{Baseline headline.} The waiting-for-benefits design estimate---0.42---treated as a natural baseline because the paper calls that design its sharpest.
    \item \textit{Design-prior class.} The simplex class $\L_{\Delta}$.
\end{itemize}

Table \ref{arxiv1:tab:ganong.application} reports the simplex-optimal headline output for all six designs and, as a narrower comparison, for the four matched-employed supplement-change designs in \citet[Table 1, rows 2--5]{ganong2024spending}. Across all six designs, the optimal headline is $\htheta_{\Delta}^{*} = 0.345$ and the optimized risk is $R_{\Delta}^{*} = 57.25$, determined by the 0.42 waiting estimate and the 0.27 September-expiration estimate. Taking $a_{B} = 0.42$ as the baseline headline gives a baseline risk of $R_{\Delta}^{B} = 226.00$; adopting the optimal headline reduces risk by 74.7\%. Restricting attention to rows 2--5 yields an optimal headline of $\htheta_{\Delta}^{*} = 0.297$ and a smaller optimized risk of $R_{\Delta}^{*} = 8.11$.

\begin{table}[H]
\centering
\caption{Simplex-optimal headline estimation: \citet{ganong2024spending}}
\label{arxiv1:tab:ganong.application}
\vspace{0.25em}
\begin{threeparttable}
\footnotesize
\setlength{\tabcolsep}{3.5pt}
\renewcommand{\arraystretch}{1.12}
\begin{tabularx}{\textwidth}{@{}p{0.25\textwidth}p{0.14\textwidth}>{\centering\arraybackslash}p{0.13\textwidth}>{\centering\arraybackslash}p{0.10\textwidth}X@{}}
\toprule
Candidate design set & Estimate range & \shortstack{Optimal \\ headline $\htheta_{\Delta}^{*}$} & \shortstack{Optimal \\ risk $R_{\Delta}^{*}$} & Active designs \\
\midrule
All six designs & 0.27--0.42 & 0.345 & 57.25 & Waiting for benefits; \$300 expiration (September states) \\
\addlinespace[6pt]
Matched-employed supplement changes, rows 2--5 & 0.27--0.35 & 0.297 & 8.11 &
\$300 expiration (June states); \$300 expiration (September states) \\
\bottomrule
\end{tabularx}
\begin{tablenotes}[flushleft]
\footnotesize
\item \textit{Notes}. The estimates and standard errors are reported in \citet[Table 1, rows 1--6]{ganong2024spending}. The ``All six designs'' calculation uses the six published estimates and standard errors. The narrower calculation uses rows 2--5, which compare supplement changes using unemployed and matched-employed households.
\end{tablenotes}
\end{threeparttable}
\end{table}

The all-design headline of 0.345 implies that about 35 cents of an additional benefit dollar is spent in the first month, while the supplement-change headline of 0.297 implies about 30 cents. Both magnitudes support the conclusion that spending responses to expanded unemployment benefits are large. This evidence is an input to the paper's broader conclusion that temporary supplements can provide substantial demand support during recessions. At the same time, the all-design risk indicates that the six designs do not tightly support one exact MPC magnitude. The same is true for the supplement-change design set, although this set yields a much lower risk. Overall, the design-specific evidence supports a stable high-MPC conclusion more tightly than it supports one exact cross-episode value.

\section{Conclusion}\label{arxiv1:sec:conclusion}
Researchers often use multiple research designs to study a single question, but lack a dedicated framework for producing a headline estimate. This paper formalizes the headline estimation problem as a decision under ambiguity about which design supplies the valid identifying link to a scalar target parameter, and proposes a conditional Gamma-minimax procedure for headline estimation under ambiguity. The procedure yields an ambiguity-optimal headline together with a measure of worst-case posterior risk given the candidate designs. The empirical applications illustrate several uses of the headline estimation procedure across different ambiguity classes. The simplex class can separate agreement about the direction of an effect from reconciliation about its exact magnitude across designs. The perturbation class can show when an existing precision-weighted average is already close to optimal. Finally, the ratio class can quantify how strongly a favored design must be privileged before an optimal headline yields low risk. 

\clearpage
\bibliography{references}
\clearpage

\appendix


\hypersetup{pageanchor=false}
\setcounter{page}{1}
\pagenumbering{arabic}
\setcounter{section}{0}
\renewcommand{\thesection}{\Alph{section}}

\begin{center}
\LARGE Supplemental Appendix to \\ ``Headline Estimation with Multiple Research Designs''
\end{center}

\begin{center}
\large Vod Vilfort
\end{center}

\vspace{1.5em}

\section{Proofs}

\subsection{Proof of Proposition \ref{arxiv1:prop:precision.weights}}\label{arxiv1:app:proof:precision.weights}
The oracle posterior risk conditional on $v = v_{k}$ is
\begin{align*}
    \int_{\M} w_{k}(\mu_{k} - Y_{k})^{2} \pi_{0}(d\mu|Y) = w_{k}\curly{(Y_{k} - Y_{k})^{2} + \sigma_{k}^{2}} = w_{k}\sigma_{k}^{2}, \quad \forall k \in [K].
\end{align*}
Because $w_{k} > 0$ for all $k \in [K]$, equalizing the above across realizations of $v$ requires $w_{k}\sigma_{k}^{2} = c$ for a constant $c > 0$, and thus $w_{k} = c/\sigma_{k}^{2}$.

\subsection{Proof of Proposition \ref{arxiv1:prop:least.favorable}}\label{arxiv1:app:proof:least.favorable}
For each $\l \in \L$, the function
\begin{align*}
    a \mapsto g(a,\l) = \sum_{k = 1}^{K} \l_{k}D_{k}(a|Y)^{2}
\end{align*}
is strongly convex with second derivative at least $2/\max_{k} \sigma_{k}^{2}$. Hence, $R_{\L}(a|Y) - 1$ is strongly convex and therefore has a unique minimizer $\htheta_{\L}^{*}(Y)$.

For the characterization, consider the compact interval $\mathcal{I} = [\min_{k}Y_{k}, \max_{k}Y_{k}]$. Projecting any $a \notin \mathcal{I}$ onto $\mathcal{I}$ weakly reduces every $D_{k}(a|Y)$ and strictly reduces at least one. The minimization can thus be restricted to $a \in \mathcal{I}$. The function $(a, \l) \mapsto g(a, \l)$ is continuous, convex in $a$, and affine in $\l$. Thus, since $\mathcal{I}$ and $\L$ are compact and convex, the minimax theorem of \citet{sion1958general} gives
\begin{align*}
    \min_{a \in \mathcal{I}} \max_{\l \in \L} g(a,\l) = \max_{\l \in \L} \min_{a \in \mathcal{I}} g(a,\l).
\end{align*}
The left side is uniquely minimized at $a=\htheta_{\L}^{*}(Y)$. Compactness, continuity, and Berge’s maximum theorem ensure that the right side attains its maximum $V$ at any least favorable $\l^{*} \in \L$. Thus,
\begin{align*}
    V = \min_{a \in \mathcal{I}} g(a,\l^{*}) \leq g(\htheta_{\L}^{*}(Y),\l^{*}) \leq \max_{\l \in \L} g(\htheta_{\L}^{*}(Y),\l) = V.
\end{align*}
Both inequalities are therefore equalities. This means $\l^{*}$ is a worst-case prior at $a=\htheta_{\L}^{*}(Y)$ and $\htheta_{\L}^{*}(Y)$ minimizes $g(a,\l^{*})$. The first-order condition for the strictly convex quadratic $g(a,\l)$ is
\begin{align*}
    \sum_{k = 1}^{K} \l_{k}\frac{a - Y_{k}}{\sigma_{k}^{2}} = 0 \iff a(\l|Y) = \frac{\sum_{k = 1}^{K} \l_{k}Y_{k}/\sigma_{k}^{2}}{\sum_{k = 1}^{K} \l_{k}/\sigma_{k}^{2}}, \quad \forall \l \in \Delta(\V).
\end{align*}
In particular, $\htheta_{\L}^{*}(Y) = a(\l^{*}|Y)$. The weighted-variance formula with weights $\l_{k}/\sigma_{k}^{2}$ yields 
\begin{align*}
    \rho(\l|Y) = \sum_{k = 1}^{K}\frac{\l_{k}}{\sigma_{k}^{2}}\curly{Y_{k} - a(\l|Y)}^{2} = \frac{\sum_{j < k} (Y_{j} - Y_{k})^{2}(\l_{j}/\sigma_{j}^{2})(\l_{k}/\sigma_{k}^{2})}{\sum_{\ell = 1}^{K} \l_{\ell}/\sigma_{\ell}^{2}}.
\end{align*}
Evaluating $g(\htheta_{\L}^{*}(Y),\l^{*}) + 1 = \rho(\l^{*}|Y) + 1$ yields the formula for $R_{\L}^{*}(Y)$. 

For $K = 2$, let $\l(p) = (p,1-p)'$. Since $\L$ is nonempty, compact, and convex, the corresponding set $\{p \in [0,1]:\l(p) \in \L\}$ of probabilities on design $1$ is an interval $\mathcal{I}_{\L} = [\underline{p}_{\L},\overline{p}_{\L}]$. For fixed $p \in [0,1]$, the first-order condition above gives
\begin{align*}
    a(p) = a(\l(p)|Y) = \frac{pY_{1}/\sigma_{1}^{2} + (1 - p)Y_{2}/\sigma_{2}^{2}}{p/\sigma_{1}^{2} + (1 - p)/\sigma_{2}^{2}}, \quad \rho(p) = \rho(\l(p)|Y) = \frac{p(1 - p)(Y_{1} - Y_{2})^{2}}{p\sigma_{2}^{2} + (1 - p)\sigma_{1}^{2}}.
\end{align*}
When $Y_{1} \neq Y_{2}$, differentiation yields
\begin{align*}
    \rho'(p) = \frac{(Y_{1} - Y_{2})^{2}\curly{\sigma_{1}^{2}(1 - p)^{2}-\sigma_{2}^{2}p^{2}}}{\curly{p\sigma_{2}^{2} + (1 - p)\sigma_{1}^{2}}^{2}}.
\end{align*}
The derivative is positive for $p < p_{\Delta}$, zero at $p = p_{\Delta}$, and negative for $p > p_{\Delta}$. Thus, the maximizer of $\rho(p)$ over $\mathcal{I}_{\L}$ is the projection $p_{\L}^{*}$ of $p_{\Delta}$ onto this interval. Evaluating $a(p)$ and $\rho(p)+1$ at $p_{\L}^{*}$ yields the displayed headline and risk formulas, and strict monotonicity on either side of $p_{\Delta}$ gives uniqueness of the least favorable prior. If $Y_{1} = Y_{2}$, then $\rho(p) = 0$ for every $p$, so every admissible design prior $\l \in \L$ is least favorable and the same formulas give headline $Y_{1} = Y_{2}$ and risk one.

\subsection{Proof of Proposition \ref{arxiv1:prop:unrestricted.simplex}}\label{arxiv1:app:proof:unrestricted.simplex}
For $c \geq 0$, there exists $a$ satisfying $D_{k}(a|Y) \leq c$ for every $k$ if and only if
\begin{align*}
    a \in \bigcap_{k = 1}^{K} \brack{Y_{k} - c\sigma_{k}, Y_{k} + c\sigma_{k}}.
\end{align*}
Given this interval-containment condition, the intersection is nonempty if and only if
\begin{align*}
    \max_{j \in [K]} \curly{Y_{j} - c\sigma_{j}} \leq \min_{k \in [K]} \curly{Y_{k} + c\sigma_{k}}.
\end{align*}
Equivalently, for every ordered pair $(j, k)$,
\begin{align*}
    Y_{j} - Y_{k} \leq c(\sigma_{j} + \sigma_{k}).
\end{align*}
The smallest feasible interval multiplier is
\begin{align*}
    c_{\Delta}^{*}(Y) = \max_{j,k \in [K]} \frac{Y_{j} - Y_{k}}{\sigma_{j} + \sigma_{k}} = \max_{j < k} \frac{\abs{Y_{j} - Y_{k}}}{\sigma_{j} + \sigma_{k}}.
\end{align*}
If $(j^{*},k^{*})$ is an active pair, relabel its members if necessary so
that $Y_{j^{*}} \geq Y_{k^{*}}$. Then
\begin{align*}
    \max_{j \in [K]} \curly{Y_{j} - c_{\Delta}^{*}(Y)\sigma_{j}} &\geq \curly{Y_{j^{*}} - \frac{Y_{j^{*}} - Y_{k^{*}}}{\sigma_{j^{*}} + \sigma_{k^{*}}}\sigma_{j^{*}}} = \frac{\sigma_{k^{*}}Y_{j^{*}} + \sigma_{j^{*}}Y_{k^{*}}}{\sigma_{j^{*}} + \sigma_{k^{*}}} \\
    \min_{k \in [K]} \curly{Y_{k} + c_{\Delta}^{*}(Y)\sigma_{k}} &\leq \curly{Y_{k^{*}} + \frac{Y_{j^{*}} - Y_{k^{*}}}{\sigma_{j^{*}} + \sigma_{k^{*}}}\sigma_{k^{*}}} = \frac{\sigma_{k^{*}}Y_{j^{*}} + \sigma_{j^{*}}Y_{k^{*}}}{\sigma_{j^{*}} + \sigma_{k^{*}}}.
\end{align*}
Thus, since $c_{\Delta}^{*}(Y)$ is the smallest feasible multiplier, the interval-containment condition gives
\begin{align*}
     \max_{k \in [K]} \curly{Y_{k} - c_{\Delta}^{*}(Y)\sigma_{k}} = \htheta_{\Delta}^{*}(Y) = \min_{k \in [K]} \curly{Y_{k} + c_{\Delta}^{*}(Y)\sigma_{k}} = \frac{\sigma_{k^{*}}Y_{j^{*}} + \sigma_{j^{*}}Y_{k^{*}}}{\sigma_{j^{*}} + \sigma_{k^{*}}}.
\end{align*}
The condition also gives $\max_{k \in [K]}D_{k}(\htheta_{\Delta}^{*}(Y)|Y) \leq c_{\Delta}^{*}(Y)$. Note further that
\begin{align*}
    \max_{k \in [K]}D_{k}(\htheta_{\Delta}^{*}(Y)|Y) \geq D_{j^{*}}(\htheta_{\Delta}^{*}(Y)|Y) = D_{k^{*}}(\htheta_{\Delta}^{*}(Y)|Y) = c_{\Delta}^{*}(Y).
\end{align*}
It therefore follows that
\begin{align*}
    c_{\Delta}^{*}(Y) = \max_{k \in [K]}D_{k}(\htheta_{\Delta}^{*}(Y)|Y) = \sqrt{R_{\Delta}^{*}(Y)-1}, \quad R_{\Delta}^{*}(Y) = \paren{\frac{\abs{Y_{j^{*}} - Y_{k^{*}}}}{\sigma_{j^{*}} + \sigma_{k^{*}}}}^{2} + 1.
\end{align*}
For $K = 2$ and $\l(p) = (p, 1 - p)'$, the simplex class generates the interval
\begin{align*}
    \mathcal{I}_{\L_{\Delta}} = \{p \in [0,1]: \l(p) \in \L_{\Delta}\} = [0,1].
\end{align*}
Proposition \ref{arxiv1:prop:least.favorable} therefore gives $p_{\L_{\Delta}}^{*} = p_{\Delta}$, and hence unique $\l^{*} = (p_{\Delta}, 1 - p_{\Delta})'$ when $Y_{1} \neq Y_{2}$.

\subsection{Proof of Proposition \ref{arxiv1:prop:epsilon.perturbation}}\label{arxiv1:app:proof:epsilon.perturbation}
For any $q \in \R^{K}$, the definition of $\L_{\epsilon}(\l^{0})$ implies
\begin{align*}
    \max_{\l \in \L_{\epsilon}(\l^{0})}\sum_{k = 1}^{K}\l_{k}q_{k}
    = \max_{\eta \in \Delta(\V)}\sum_{k = 1}^{K}\curly{(1 - \epsilon)\l_{k}^{0} + \epsilon\eta_{k}}q_{k}
    = (1 - \epsilon)\sum_{k = 1}^{K}\l_{k}^{0}q_{k} + \epsilon\max_{k \in [K]}q_{k}.
\end{align*}
Substituting $q_{k} = D_{k}(a|Y)^{2}$ and adding one gives the formula for $R_{\epsilon}(a|Y)$. By Proposition \ref{arxiv1:prop:least.favorable}, $\htheta_{\epsilon}^{*}(Y)$ exists uniquely and a least favorable prior $\l_{\epsilon}^{*}$ can be computed as stated.

I now show continuity and monotonicity of $\epsilon \mapsto \htheta_{\epsilon}^{*}(Y) = a_{\epsilon} = \arg\min_{a \in \R}F_{\epsilon}(a)$, where
\begin{align*}
    F_{\epsilon}(a) = (1 - \epsilon)B(a) + \epsilon M(a), \quad B(a) = \sum_{k = 1}^{K}\l_{k}^{0}D_{k}(a|Y)^{2}, \quad M(a) = \max_{k \in [K]}D_{k}(a|Y)^{2}.
\end{align*}
Note $a_{\epsilon} \in [\min_{k}Y_{k},\max_{k}Y_{k}]$. On this compact interval, $(a,\epsilon) \mapsto F_{\epsilon}(a)$ is jointly continuous and $F_{\epsilon}$ has unique minimizer $a_{\epsilon}$, so Berge's maximum theorem implies continuity of $\epsilon \mapsto a_{\epsilon}$. If $a_{0} = a_{1}$, this common value uniquely minimizes both $B$ and $M$, and hence $F_{\epsilon}$ for every $\epsilon$. Suppose that $a_{0} < a_{1}$. Since $B$ and $M$ are strongly convex with unique minimizers $a_{0}$ and $a_{1}$, respectively, no $F_{\epsilon}$ can be minimized outside $[a_{0},a_{1}]$: if $a < a_{0}$, replacing $a$ with $a_{0}$ lowers both $B$ and $M$, while if $a > a_{1}$, replacing $a$ with $a_{1}$ lowers both. On $[a_{0},a_{1}]$, $B$ is strictly increasing and $M$ is strictly decreasing, so $M-B$ is strictly decreasing. Hence,
\begin{align*}
    F_{\epsilon}(a) = B(a) + \epsilon\curly{M(a)-B(a)}
\end{align*}
has decreasing differences in $(a,\epsilon)$. Topkis's theorem for minimization, together with uniqueness, implies $\epsilon \mapsto a_{\epsilon}$ is nondecreasing. The case $a_{0} > a_{1}$ is symmetric and implies nonincreasing $\epsilon \mapsto a_{\epsilon}$.
The $a_{\epsilon}$ endpoints follow from Proposition \ref{arxiv1:prop:least.favorable} for $\L_{0}(\l^{0}) = \{\l^{0}\}$ and Proposition \ref{arxiv1:prop:unrestricted.simplex} for $\L_{1}(\l^{0}) = \L_{\Delta}$.

For $K = 2$ and $\l(p) = (p, 1 - p)'$, the perturbation class generates the interval
\begin{align*}
    \mathcal{I}_{\epsilon} = \curly{p \in [0,1]: \l(p) \in \L_{\epsilon}(\l^{0})} = [(1 - \epsilon)p^{0},(1 - \epsilon)p^{0} + \epsilon].
\end{align*}
Proposition \ref{arxiv1:prop:least.favorable} therefore gives $p_{\epsilon}^{*}$ as the projection of $p_{\Delta}$ onto $\mathcal{I}_{\epsilon}$ and yields the headline and risk formulas. If $p_{\Delta} \in \mathcal{I}_{\epsilon}$, then $p_{\epsilon}^{*} = p_{\Delta}$ so that $\htheta_{\epsilon}^{*}(Y) = \htheta_{\Delta}^{*}(Y)$.

\subsection{Proof of Proposition \ref{arxiv1:prop:preferred.design}}\label{arxiv1:app:proof:preferred.design}
Every $\l \in \L_{r}$ satisfies $\l_{1} > 0$. For $k \neq 1$, define
$p_{k} = r\l_{k}/\l_{1} \in [0,1]$. The simplex constraint gives
\begin{align*}
    \l_{1} = \frac{r}{r+\sum_{k \neq 1}p_{k}}, \quad \l_{k} = \frac{p_{k}}{r+\sum_{j \neq 1}p_{j}}, \quad k \neq 1.
\end{align*}
Conversely, these formulas map every $p \in [0,1]^{K-1}$ to a design prior in $\L_{r}$. The support function for $\L_{r}$ is given by
\begin{align*}
    \max_{\l \in \L_{r}}\sum_{k=1}^{K}\l_{k}q_{k} = \max_{p \in [0,1]^{K-1}}H(p), \quad H(p) = \frac{ rq_{1} + \sum_{k \neq 1}p_{k}q_{k}}{r + \sum_{k \neq 1}p_{k}}, \quad q \in \R^{K}.
\end{align*}
For each $k \neq 1$, $H(p)$ can be expressed in the form $(A+p_{k}q_{k})/(B+p_{k})$, whose partial derivative with respect to $p_{k}$ has a sign that does not depend on $p_{k}$. A maximizer of $H(p)$ can therefore be chosen with $p_{k} \in \{0,1\}$ for each $k \neq 1$. That is, $\max_{p \in [0,1]^{K-1}}H(p) = \max_{p \in \{0,1\}^{K-1}}H(p)$.

Let $S = \{k \neq 1: p_{k}=1\}$ and $m = \abs{S}$. The corresponding prior $\l^{S}$ places probability $r/(r+m)$ on design $1$, probability $1/(r+m)$ on $k \in S$, and probability zero on other designs, yielding
\begin{align*}
    H(p) = \frac{rq_{1} + \sum_{k \in S}q_{k}}{r+m}, \quad \l_{1} = \frac{r}{r+m}, \quad \l_{k} = \frac{\I{k \in S}}{r+m}, \quad k \neq 1, \quad m = \abs{S}.
\end{align*}
For fixed $m$, this is maximized by choosing the $m$ largest alternative coordinates. Hence,
\begin{align*}
    \max_{\l \in \L_{r}} \sum_{k = 1}^{K}\l_{k}q_{k} = \max_{m \in \{0,\ldots,K-1\}} \frac{rq_{1} + \sum_{j = 1}^{m}q_{-1,(j)}}{r+m},
\end{align*}
where $q_{-1,(1)} \geq \cdots \geq q_{-1,(K-1)}$
is a decreasing rearrangement of $\curly{q_{k}: k \neq 1}$. For $q_{k} = D_{k}(a|Y)^{2}$, substituting into the above display and adding one gives the stated formula for $R_{r}(a|Y)$. By Proposition \ref{arxiv1:prop:least.favorable}, $\htheta_{r}^{*}(Y)$ exists uniquely and a least favorable prior $\l_{r}^{*}$ can be computed as stated.

For each $S \subseteq [K] \setminus \{1\}$, consider the priors $\l^{S}$ described above. Because the above equality for $\max_{\l \in \L_{r}}\l'q$ holds for every $q \in \R^{K}$, support function uniqueness yields the polytope
\begin{align*}
    \L_{r} = \operatorname{conv}\curly{\l^{S}: S \subseteq [K]\setminus\{1\}}.
\end{align*}
Consequently, since $\l \mapsto \l'q$ is affine, the maximizing design priors at $a = \htheta_{r}^{*}(Y)$ are the convex hull of the subset of maximizing extreme priors $\l^{S}$. Proposition \ref{arxiv1:prop:least.favorable} gives a least favorable prior in this convex hull. If the maximizing $m^{*}$ and $S^{*}$ are unique, the least favorable prior is $\l^{S^{*}}$.

For $r_{2} \geq r_{1}$, $\L_{r_{2}} \subseteq \L_{r_{1}}$, so $R_{r}^{*}(Y)$ is nonincreasing. To establish its limit, evaluate the criterion at $a = Y_{1}$. The extreme-prior formula gives
\begin{align*}
    R_{r}^{*}(Y) - 1 \leq \max_{m \in \{0,\ldots,K - 1\}} \frac{\sum_{j = 1}^{m}D_{-1,(j)}(Y_{1}|Y)^{2}}{r + m} \leq \frac{K - 1}{r}\max_{k \neq 1} D_{k}(Y_{1}|Y)^{2} \to 0.
\end{align*}
The dogmatic prior on the favored design belongs to $\L_{r}$ for every $r$, so
\begin{align*}
    D_{1}(\htheta_{r}^{*}(Y)|Y)^{2} \leq R_{r}^{*}(Y) - 1 \to 0.
\end{align*}
Therefore, $\htheta_{r}^{*}(Y) \to Y_{1}$.

For $K = 2$ and $\l(p) = (p, 1 - p)'$, the ratio class generates the interval
\begin{align*}
    \mathcal{I}_{r} = \curly{p \in [0,1]: \l(p) \in \L_{r}} = [p_{r}, 1].
\end{align*}
Proposition \ref{arxiv1:prop:least.favorable} therefore gives $p_{r}^{*}$ as the
projection of $p_{\Delta}$ onto $\mathcal{I}_{r}$ and yields the headline
and risk formulas. If $p_{\Delta} \in \mathcal{I}_{r}$, then $p_{r}^{*} = p_{\Delta}$ so that $\htheta_{r}^{*}(Y) = \htheta_{\Delta}^{*}(Y)$.

\subsection{Proof of Proposition \ref{arxiv1:prop:plugin.local}}\label{arxiv1:app:proof:plugin.local}
I first show continuity of functions that will also be used in the proof of Proposition \ref{arxiv1:prop:persistent.disagreement}. For $\mathbb{K}_{\Delta}$ the set of nonempty compact subsets of $\Delta(\V)$, $c \in \R^{K}$, $s \in (0,\infty)^{K}$, and $A \in \mathbb{K}_{\Delta}$, define
\begin{align*}
    \mathcal{Q}(a;c,s,A)
    = \max_{\l \in A}\sum_{k = 1}^{K}\l_{k}\frac{(c_{k} - a)^{2}}{s_{k}}, \quad \mathcal{T}(c,s,A) &= \arg\min_{a \in \R}\mathcal{Q}(a;c,s,A), \\
    \mathcal{M}(c,s,A) &= \min_{a \in \R}\mathcal{Q}(a;c,s,A).
\end{align*}
For every $\l \in A$, the criterion inside the maximum is strongly convex in $a$, with second derivative
\begin{align*}
    2\sum_{k = 1}^{K}\frac{\l_{k}}{s_{k}} \geq \frac{2}{\max_{k \in [K]}s_{k}}.
\end{align*}
The maximum over $\l \in A$ therefore preserves strong convexity and so $\mathcal{T}(c,s,A)$ exists uniquely. Moreover, projecting any $a$ outside of $[\min_{k}c_{k},\max_{k}c_{k}]$ onto this interval reduces every squared discrepancy. Hence,
\begin{align*}
    \mathcal{T}(c,s,A) \in \brack{\min_{k \in [K]}c_{k},\max_{k \in [K]}c_{k}}.
\end{align*}

To show continuity of $\mathcal{Q}$, define $q(a;c,s) \in \R^{K}$ with continuous components
\begin{align*}
    (a,c,s) \mapsto q_{k}(a;c,s) = \frac{(c_{k} - a)^{2}}{s_{k}}.
\end{align*}
For support function $h_{A}(q) = \max_{\l \in A}\l'q$, Hausdorff distance $d_{\mathrm{H}}$, $A,B \in \mathbb{K}_{\Delta}$, and $q,\widetilde{q} \in \R^{K}$,
\begin{align*}
    \abs{h_{A}(q) - h_{B}(\widetilde{q})}
    \leq \max_{\l \in A}\norm{\l}_{2}\norm{q - \widetilde{q}}_{2} + \norm{\widetilde{q}}_{2}d_{\mathrm{H}}(A,B) \leq \norm{q - \widetilde{q}}_{2} + \norm{\widetilde{q}}_{2}d_{\mathrm{H}}(A,B).
\end{align*}
Hence, $(a,c,s,A) \mapsto \mathcal{Q}(a;c,s,A) = h_{A}(q(a;c,s))$ is continuous.

To show continuity of $\mathcal{T}$ and $\mathcal{M}$, consider any sequence $(c_{j},s_{j},A_{j}) \to (c,s,A)$ and let
\begin{align*}
    \mathcal{Q}_{j}(a) = \mathcal{Q}(a;c_{j},s_{j},A_{j}), \quad
    \mathcal{Q}_{0}(a) = \mathcal{Q}(a;c,s,A), \quad (c_{j},s_{j},A_{j}) \to (c,s,A), \quad j \to \infty.
\end{align*}
Because $c_{j} \to c$, there is a compact interval $I \subseteq \R$ containing $[\min_{k}c_{j,k},\max_{k}c_{j,k}]$ for large $j$ and also containing $[\min_{k}c_{k},\max_{k}c_{k}]$. Thus, all corresponding minimizers lie in $I$. Since $s \in (0,\infty)^{K}$ and $s_{j} \to s$, the denominators are bounded away from zero for large $j$. It follows that
\begin{align*}
    \sup_{a \in I}\norm{q(a;c_{j},s_{j}) - q(a;c,s)}_{2} \to 0, \quad
    \sup_{a \in I}\norm{q(a;c,s)}_{2} < \infty.
\end{align*}
Combining these bounds with $d_{\mathrm{H}}(A_{j},A) \to 0$ and the support-function inequality gives
\begin{align*}
    \sup_{a \in I}\abs{\mathcal{Q}_{j}(a) - \mathcal{Q}_{0}(a)} \to 0.
\end{align*}
Let $a_{j} = \mathcal{T}(c_{j},s_{j},A_{j})$ and $a_{0} = \mathcal{T}(c,s,A)$. Take an arbitrary subsequence of $(c_{j},s_{j},A_{j})$, and hence of $a_{j}$. Compactness of $I$ gives a further subsequence, again indexed by $j$, such that $a_{j} \to \Bar{a} \in I$. By optimality of $a_{j}$, $\mathcal{Q}_{j}(a_{j}) \leq \mathcal{Q}_{j}(a)$ $\forall a \in I$. Uniform convergence on $I$ and continuity of $\mathcal{Q}_{0}$ imply
\begin{align*}
    \mathcal{Q}_{0}(\Bar{a}) \leq \mathcal{Q}_{0}(a), \quad \forall a \in I.
\end{align*}
The global minimizer of $\mathcal{Q}_{0}$ is in $I$, so $\Bar{a}$ is the global minimizer. Uniqueness gives $\Bar{a} = a_{0}$. Since every subsequence has a further subsequence converging to $a_{0}$, the full sequence satisfies $a_{j} \to a_{0}$. This proves continuity of $\mathcal{T}$. Finally,
\begin{align*}
    \abs{\mathcal{M}(c_{j},s_{j},A_{j}) - \mathcal{M}(c,s,A)}
    &\leq \abs{\mathcal{Q}_{j}(a_{j}) - \mathcal{Q}_{0}(a_{j})}
    + \abs{\mathcal{Q}_{0}(a_{j}) - \mathcal{Q}_{0}(a_{0})} \to 0,
\end{align*}
which proves continuity of $\mathcal{M}$.

Now for the main claim, consider
\begin{align*}
    Z_{n} = \sqrt{n}(Y_{n} - \theta\1), \quad
    s_{n} = \paren{n\tsigma_{1,n}^{2},\ldots,n\tsigma_{K,n}^{2}}', \quad
    s^{0} = \paren{\sigma_{1}^{2},\ldots,\sigma_{K}^{2}}'.
\end{align*}
The definitions of the plug-in headline and plug-in risk imply
\begin{align*}
    \sqrt{n}(\htheta_{n,\L}^{*} - \theta)
    = \mathcal{T}(Z_{n},s_{n},\L), \quad
    R_{n,\L}^{*} = \mathcal{M}(Z_{n},s_{n},\L) + 1.
\end{align*}
Assumption \ref{arxiv1:ass:local.gaussian} and Slutsky's theorem give
\begin{align*}
    (Z_{n},s_{n}) \to[d] (Z,s^{0}), \quad Z \sim N(b,\Sigma).
\end{align*}
Continuity of $\mathcal{T}$ and $\mathcal{M}$ and the continuous mapping theorem (CMT) therefore yield
\begin{align*}
    \sqrt{n}(\htheta_{n,\L}^{*} - \theta) \to[d] \mathcal{T}(Z,s^{0},\L) = \htheta_{\L}^{*}(Z), \quad R_{n,\L}^{*} \to[d] \mathcal{M}(Z,s^{0},\L) + 1 = R_{\L}^{*}(Z).
\end{align*}

\subsection{Proof of Proposition \ref{arxiv1:prop:persistent.disagreement}}\label{arxiv1:app:proof:persistent.disagreement}
Use the continuous functions $\mathcal{T}$ and $\mathcal{M}$ defined in the proof of Proposition \ref{arxiv1:prop:plugin.local}, and let
\begin{align*}
    s_{n} = \paren{n\tsigma_{1,n}^{2},\ldots,n\tsigma_{K,n}^{2}}', \quad
    s^{0} = \paren{\sigma_{1}^{2},\ldots,\sigma_{K}^{2}}'.
\end{align*}
Multiplying the plug-in objective by $1/n$ does not change its minimizer. Consequently,
\begin{align*}
    \htheta_{n,\L}^{*} = \mathcal{T}(Y_{n},s_{n},\L), \quad \frac{R_{n,\L}^{*}}{n} = \mathcal{M}(Y_{n},s_{n},\L) + \frac{1}{n}.
\end{align*}
Continuity of $\mathcal{T}$ and $\mathcal{M}$, $(Y_{n},s_{n}) \to[p] (\mu^{0},s^{0})$, and CMT therefore imply
\begin{align*}
    \htheta_{n,\L}^{*}
    \to[p] \mathcal{T}(\mu^{0},s^{0},\L)
    = \theta_{\L}^{\mathrm{lim}}(\mu^{0}), \quad \frac{R_{n,\L}^{*}}{n} \to[p] \mathcal{M}(\mu^{0},s^{0},\L)
    = R_{\L}^{\mathrm{lim}}(\mu^{0}).
\end{align*}
If $R_{\L}^{\mathrm{lim}}(\mu^{0}) = 0$, there exists an $a \in \R$ for which
\begin{align*}
    \max_{\l \in \L}\sum_{k = 1}^{K}\l_{k}\frac{(\mu_{k}^{0} - a)^{2}}{\sigma_{k}^{2}} = 0.
\end{align*}
Every summand is nonnegative. If every design receives positive probability under at least one prior in $\L$, the above equality requires $\mu_{k}^{0} = a$ for every $k \in [K]$. Therefore, if the coordinates of $\mu^{0}$ are not all equal, then $R_{\L}^{\mathrm{lim}}(\mu^{0}) > 0$.

\end{document}